\documentclass[aps,pra,10pt,letterpaper,twocolumn,tightenlines,superscriptaddress,notitlepage,longbibliography]{revtex4-2}
\usepackage{qcircuit}
\usepackage[dvips]{graphicx}
\usepackage{amsmath,amssymb,amsthm,mathrsfs,amsfonts,dsfont}
\usepackage{mathtools}
\usepackage{subfigure, epsfig}
\usepackage{braket}
\usepackage{bm}
\usepackage{enumerate}
\usepackage{xcolor}
\usepackage{comment}
\usepackage{scalefnt}
\usepackage[colorlinks = true, 
citecolor = blue,
linkcolor = blue,
urlcolor = blue]{hyperref}

\usepackage{pagecolor}
\usepackage[T1]{fontenc}
\usepackage{lmodern}
\usepackage[utf8]{inputenc}
\usepackage{microtype}

\usepackage{tikzscale}
\usepackage{pgfplots}
\pgfplotsset{compat=newest}
\usetikzlibrary{plotmarks}
\usetikzlibrary{arrows.meta}
\usetikzlibrary{external}
\usepgfplotslibrary{patchplots}
\usepackage{grffile}
\usepackage{enumitem}

\DeclareMathOperator{\tr}{Tr}
\DeclareMathOperator{\Tr}{Tr}

\DeclarePairedDelimiter\floor{\lfloor}{\rfloor}

\newcommand{\coleq}{\mathrel{\mathop:}\nobreak\mkern-1.2mu=}
\newcommand{\pe}{\mathrel{\mathop+}\nobreak\mkern-1.2mu=}

\newcommand{\mc}{\mathcal}

\newcommand{\mr}{\mathrm}
\newcommand{\mbb}{\mathbb}

\newcommand{\expval}[1]{{\langle #1 \rangle}}
\newcommand{\ketbra}[2]{{\vert #1 \rangle \langle #2 \vert}}
\newcommand{\lket}[1]{\vert #1 \rangle\!\rangle}
\newcommand{\lbra}[1]{\langle\!\langle #1 \vert}
\newcommand{\lbraket}[2]{\langle\!\langle #1 \vert #2 \rangle\!\rangle}
\newcommand{\lketbra}[2]{\vert #1 \rangle\!\rangle\langle\!\langle #2 \vert}
\newcommand{\ptm}{{\mathrm{PTM}}}

\newtheorem{theorem}{Theorem}

\newtheorem{lemma}[theorem]{Lemma}%

\newtheorem{corollary}[theorem]{Corollary}%
\newtheorem{definition}{Definition}

\definecolor{applegreen}{rgb}{0.55, 0.71, 0.0}

\newcommand{\comments}[1]{}

\usepackage{algorithm}  
\usepackage[noend]{algpseudocode}  
\newcommand{\algorithmfootnote}[2][\footnotesize]{%
  \let\old@algocf@finish\@algocf@finish%
  \def\@algocf@finish{\old@algocf@finish%
    \leavevmode\rlap{\begin{minipage}{\linewidth}
    #1#2
    \end{minipage}}%
  }%
}
\usepackage{xparse}
\NewDocumentCommand{\LeftComment}{s m}{%
  \Statex \IfBooleanF{#1}{\hspace*{\ALG@thistlm}}\(\triangleright\) #2}
\algnewcommand{\LineComment}[1]{\Statex // #1}

\begin{document}

\title{Quantum advantages for Pauli channel estimation}
\author{Senrui Chen}
\affiliation{Pritzker School of Molecular Engineering, The University of Chicago, Illinois 60637, USA}
\author{Sisi Zhou}
\affiliation{Pritzker School of Molecular Engineering, The University of Chicago, Illinois 60637, USA}
\affiliation{Institute for Quantum Information and Matter, California Institute of Technology, Pasadena, CA 91125, USA}
\author{Alireza Seif}
\affiliation{Pritzker School of Molecular Engineering, The University of Chicago, Illinois 60637, USA}
\author{Liang Jiang}
\affiliation{Pritzker School of Molecular Engineering, The University of Chicago, Illinois 60637, USA}
\date{\today}
\begin{abstract}

We show that entangled measurements provide an exponential advantage in sample complexity for Pauli channel estimation,
which is both a fundamental problem and a practically important subroutine for benchmarking near-term quantum devices.
The specific task we consider is to simultaneously learn all the eigenvalues of an $n$-qubit Pauli channel to $\pm\varepsilon$ precision.
We give an estimation protocol with an $n$-qubit ancilla that succeeds with high probability using only $\mc O(n/\varepsilon^{2})$ copies of the Pauli channel, while prove that any ancilla-free protocol (possibly with adaptive control and channel concatenation) would need at least $\Omega(2^{n/3})$ rounds of measurement.
We further study the advantages provided by a small number of ancillas. For the case that a $k$-qubit ancilla ($k\le n$) is available, we obtain a sample complexity lower bound of $\Omega(2^{(n-k)/3})$ for any non-concatenating protocol, and a stronger lower bound of $\Omega(n2^{n-k})$ for any non-adaptive, non-concatenating protocol,
which is shown to be tight.
We also show how to apply the ancilla-assisted estimation protocol to a practical quantum benchmarking task in a noise-resilient and sample-efficient manner, given reasonable noise assumptions.
Our results provide a practically-interesting example for quantum advantages in learning and also bring new insight for quantum benchmarking.

\end{abstract}
\maketitle

\begin{figure*}[t]
	\centering
	\includegraphics[width = 0.80\textwidth]{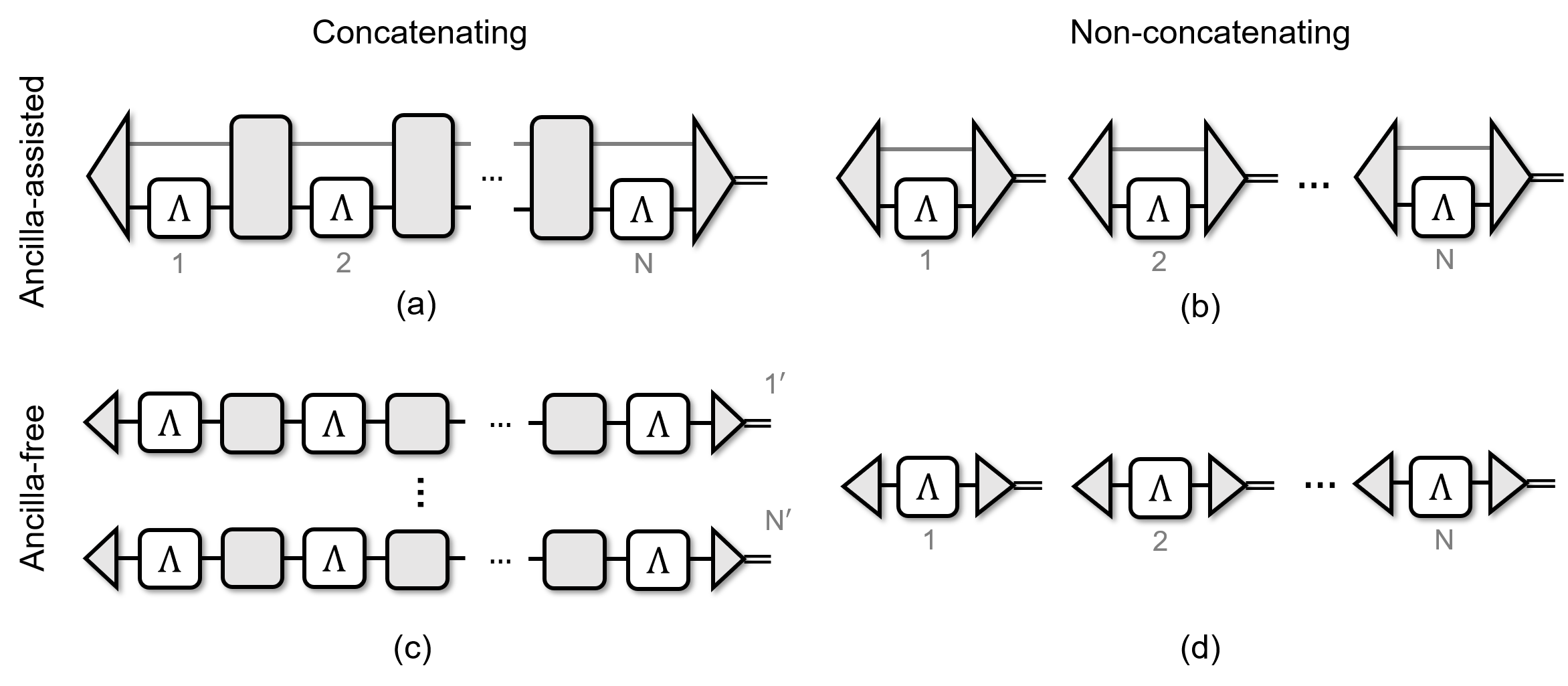}
	\caption{Different measurement models on $N$ copies of an unknown channel $\Lambda$. 
	The gray triangles denote state preparation and measurement, and the gray boxes denote some known processing channels.
	(a) \emph{Fully entangled measurement}: the most general way to measure a quantum channel where arbitrarily large entanglement and quantum memory is allowed; (b) \emph{Ancilla-assisted non-concatenating measurement}: For each sample of $\Lambda$, one input some entangled state and conduct an entangled measurement. The measurement must be completely destructive, which means no quantum memory is allowed; (c) \emph{Ancilla-free concatenating measurement}: One is allowed to sequally apply multiple copies of $\Lambda$ and some processing channels before applying a single round of measurement, with no ancilla allowed. Here $N'$ denotes the number of measurement rounds which is no larger than the number of samples; (d) \emph{Un-entangled measurement}: Neither ancilla nor concatenation is allowed.
	Additionally, measurement models (b), (c), (d) can be either \emph{adaptive} or \emph{non-adaptive}. If the measurement setting at a certain round depends on previous measurement outcomes, it is an adaptive protocol. Otherwise, it is non-adaptive. 
	}
	\label{fig:model}
\end{figure*}

One important challenge for the Noisy Intermediate-Scale Quantum (NISQ) era~\cite{preskill2018quantum} is to demonstrate quantum advantages on near-term devices.
Recent works have made groundbreaking progress towards demonstrating \emph{quantum computational advantages}~\cite{harrow2017quantum,arute2019quantum,zhong2020quantum,wu2021strong}, which means quantum computers can efficiently solve certain computational problems outside the reach of the most advanced classical computers.
However, computation is not the only aspect quantum computers can achieve meaningful advantages. 
Learning is yet another possibility, 
where the basic question is whether quantum computers (with resources such as \emph{quantum memory} and \emph{entangled measurements}) can help us learn certain properties of a physical system more efficiently.
This kind of quantum advantages have been explored by several recent works~\cite{bubeck2020entanglement,aharonov2021quantum,huang2021information,chen2021toward,rossi2021quantum} with positive examples including mixedness testing~\cite{bubeck2020entanglement}, unitarity testing~\cite{aharonov2021quantum}, Pauli expectation values estimation~\cite{huang2021information}, \textit{etc}.
However, all these quantum advantages proposals so far either focus on artificial problems or have no noise-resilient implementation.
It is thus highly desirable to identify a practically-interesting learning task that can be used to demonstrate a robust quantum advantage in the NISQ era.

A particularly interesting type of learning tasks, known as \emph{quantum benchmarking}~\cite{eisert2020quantum}, aims at characterizing the noise on a quantum device. This is yet another challenge for the NISQ era which is crucial to building better quantum hardware.
Pauli noise is one of the most important quantum noise models: On the one hand, it can describe a wide range of incoherent noise, including dephasing, deplorizing, bit-flip, etc. On the other hand, the recently developed ``randomized compiling'' technique~\cite{wallman2016noise,hashim2020randomized} can tailor any noise model in a universal gate set into Pauli noise. 
Many benchmarking protocols also rely on twirling the noise into a Pauli channel before extracting any information~\cite{erhard2019characterizing,magesan2011scalable,harper2020efficient,liu2021benchmarking}.
Because of the important role played by Pauli channels, it is of great interest to study the estimation of these objects in an efficient and practical way. 
Despite a long line of research~\cite{fujiwara2003quantum,hayashi2010quantum,chiuri2011experimental,ruppert2012optimal,collins2013mixed,flammia2020efficient,harper2021fast,flammia2021pauli,mohseni2008quantum,flammia2021pauli}, the ultimate sample complexity for Pauli channel estimation has not yet been fully characterized.

In this work, we demonstrate an {\textit{exponential}} quantum advantage for Pauli channel estimation. We show that, for the task of estimating the eigenvalues of an $n$-qubit Pauli channel (\textit{i.e.}, \emph{Pauli eigenvalues}) to additive error $\varepsilon$ in $l_\infty$ distance, there exists a measurement protocol assisted with an $n$-qubit ancilla that succeeds with high probability using $\mc O(n/\varepsilon^2)$ samples, while any ancilla-free protocol (possibly with adaptive control and channel concatenation) would require at least $\Omega(2^{n/3})$ rounds of measurements.
As a byproduct, this provides a lower bound for randomized benchmarking (RB)~\cite{magesan2011scalable,knill2008randomized} based Pauli noise estimation protocol, resolving an open problem raised in~\cite{flammia2020efficient}.

We then study the sample efficiency advantages provided by a restricted amount of ancilla.
While an ancilla larger than $n$ qubits will not help further improve the efficiency,
given a $k$-qubit $(0\le k\le n)$ ancilla, we obtain a lower bound of $\Omega(2^{(n-k)/3})$ for any non-concatenating protocol, and a stronger lower bound of $\Omega(n2^{n-k})$ for non-adaptive and non-concatenating protocols.
The latter is shown to be tight by an explicitly constructed protocol (see Algorithm~\ref{alg:main}).

Finally, we show how to apply the ancilla-assisted estimation protocol in the practical task of benchmarking Pauli gates.
Inspired by RB-type methods~\cite{flammia2020efficient, knill2008randomized,magesan2011scalable}, we design a protocol that is both robust against state-preparation-and-measurement (SPAM) errors and exponentially more sample-efficient than any ancilla-free scheme, under reasonable noise assumptions.
This protocol can be used to experimentally demonstrate a robust and practical quantum advantage on NISQ devices.

\textit{Preliminaries.{\textemdash}}
For an $n$-qubit Hilbert space, define ${\sf P}^n$ to be the Pauli group modulo the non-physical phase.
${\sf P}^n$ is an Abelian group isomorphic to $\mbb Z^{2n}_2$, so we can use elements of $\mbb Z^{2n}_2$ to uniquely label elements of ${\sf P}^n$. 
Specifically, we view every $a\in\mbb Z^{2n}_2$ as a $2n$-bit string $a={a_{x,1}a_{z,1}a_{x,2}a_{z,2}\cdots a_{x,n}a_{z,n}}$ corresponding to the Pauli operator
$$
	P_a = \otimes_{k=1}^n i^{a_{x,k}a_{z,k}} X^{a_{x,k}}Z^{a_{z,k}}
$$ 
where the phase is chosen to ensure Hermiticity.
We also define a sympletic inner product $\expval{\cdot,\cdot}$ within $\mbb Z_{2}^{2n}$ as 
$$
	\expval{a,b} = \sum_{k=1}^n (a_{x,k}b_{z,k}+a_{z,k}b_{x,k}) \mod{2}.
$$
One can verify that $P_aP_b = (-1)^\expval{a,b}P_bP_a$~\cite{nielsen2002quantum}.

An $n$-qubit Pauli channel $\Lambda$ is a quantum channel of the following form
\begin{equation}
	\Lambda(\cdot) = \sum_{a\in\mbb Z_2^{2n}}p_a P_a(\cdot)P_a,
\end{equation}
where $\bm{p}\coleq \{p_a\}_a$ is called the \emph{Pauli error rates}.
An important property of Pauli channels is that their eigen-operators are exactly the $4^n$ Pauli operators. Thus, an alternative expression for $\Lambda$ is
\begin{equation}
	\Lambda(\cdot) = \frac{1}{2^n}\sum_{b\in\mbb Z_2^{2n}}\lambda_b\Tr(P_b(\cdot))P_b,
\end{equation}
where $\bm\lambda \coleq \{\lambda_b\}_b$ is called the \emph{Pauli eigenvalues}~\cite{flammia2020efficient,flammia2021pauli}. 
These two sets of parameters, $\bm p$ and $\bm \lambda$, are related by the Walsh-Hadamard transform
\begin{equation}
	\begin{aligned}
		\lambda_b = \sum_{a\in\mbb Z_2^{2n}}p_a(-1)^\expval{a,b},\quad
		p_a = \frac{1}{4^n}\sum_{b\in\mbb Z_2^{2n}}\lambda_b(-1)^\expval{a,b}.
	\end{aligned}
\end{equation}

\noindent Both $\bm p$ and $\bm \lambda$ are physically interesting parameters:
The Pauli error rates are directly related to the error thresholds in fault-tolerant quantum computation~\cite{shor1996fault,aharonov2008fault} and have been the quantities of interest for many quantum benchmarking protocols~\cite{flammia2020efficient,harper2020efficient,harper2021fast,flammia2021pauli,liu2021benchmarking};
The Pauli eigenvalues quantify how well a Pauli observable is preserved through the noise channel (hence also known as \emph{Pauli fidelities}) and have applications in quantum error mitigation (see \textit{e.g.}~\cite{chen2020robust}).
In this work, we will focus on the estimation for $\bm\lambda$.

\textit{Comparison of different measurement models.{\textemdash}}
To study the advantages provided by different kinds of quantum resources,
we categorize measurement strategies into ancilla-assisted \textit{vs.} ancilla-free, concatenating \textit{vs.} non-concatenating, and adaptive \textit{vs.} non-adaptive measurements.
See Fig.~\ref{fig:model} and explanations in the captions.
We remark that, recent works on quantum advantages in property learning~\cite{huang2021information,aharonov2021quantum} have been focusing on the difference between the full-fledged entangled measurement Fig.~\ref{fig:model}(a) and the un-entangled measurement Fig.~\ref{fig:model}(d). 
Here, we introduce two intermediate measurement models Fig.~\ref{fig:model}(b) and Fig.~\ref{fig:model}(c) to separately study the role of ancilla and concatenation.
There are also practical reasons to care about those intermediate models. As an example, the ancilla-free concatenating measurement model Fig.~\ref{fig:model}(c) exactly describes most existing RB protocols (see \textit{e.g.}~\cite{helsen2020general}), including the RB-type Pauli channel estimation protocol in Ref.~\cite{flammia2020efficient}. 
Indeed, we show that it is the ancilla that provides an exponential advantage in sample complexity for Pauli channel estimation, while concatenation does not help improve the sample efficiency.
These results further advances our understanding of quantum advantages in property learning~\cite{huang2021information,aharonov2021quantum,bubeck2020entanglement,chen2021toward}.

Regarding these measurement models,
one further question to ask is whether a restricted amount of ancilla can provide any sample efficiency advantages.
In this work, we will characterize the advantages provided by a $k$-qubit ancilla ($0\le k\le n$) for the non-concatenating measurement models. 
This kind of quantitative trade-off relations between quantum resources and sample efficiency has not been explored in previous works~\cite{huang2021information,aharonov2021quantum,bubeck2020entanglement,chen2021toward} and may potentially lead to a new resource-theoretical interpretation of quantum entanglement~\cite{chitambar2019quantum}.

\textit{Upper bounds.{\textemdash}}
Our goal is to estimate the Pauli eigenvalues $\bm\lambda$ of an $n$-qubit Pauli channel $\Lambda$ to $\varepsilon$ precision in $l_\infty$ distance, \textit{i.e.}, estimating each $\lambda_a$ to $\varepsilon$ additive error.

When an $n$-qubit ancillary system is available, a simple protocol is as follows:
prepare $n$ Bell pairs, input one qubit from each pair to the Pauli channel, and apply a Bell measurement on the output.
Since the Pauli channel can be viewed as randomly applying one of the $4^n$ Pauli operators $P_a$ with probability $p_a$,
and each $P_a$ is mapped to a unique measurement outcome~\footnote{
This can be viewed as a consequence of the well-known \emph{superdense coding} protocol~\cite{bennett1992communication}.
}, we are effectively sampling from the probability distribution $\bm p$. 
The sample of $\bm p$ can then be used to construct an estimator for $\bm \lambda$ according to the Walsh-Hadamard transform. 
As shown in Theorem~\ref{th:upper}, this protocol has sample complexity $\mc O(n)$; 
on the other hand, when no ancilla is allowed, there is no way to sample from $\bm p$ and simultaneously estimate all elements of $\bm\lambda$ from a single measurement setting. Intuitively, this is what makes the task difficult.

In the following, we give a unified estimation protocol using $k$ ancilla qubits for $0\le k\le n$.
Roughly speaking, we divide the $n$ qubits of the Pauli channel into two disjoint subsystems containing $k$ and $n-k$ qubits, respectively, and deal with them separately. For the first subsystem, we introduce a $k$-qubit ancilla, input $k$ Bell pairs to both systems, and apply a Bell measurement; For the second subsystem, we will use \emph{stabilizer states} input and \emph{syndrome measurements} to be defined later. 
The measurement scheme is depicted in Fig.~\ref{fig:alg}. 

A rigorous description of the protocol is given in Algorithm~\ref{alg:main}.
Here, $\ket{\Psi_v}$ represents the Bell states on the $2k$-qubit subsystem defined as
\begin{equation}
	\ket{\Psi_v} = P_v\otimes I \ket{\Psi^+},\quad \ket{\Psi^+} = \frac{1}{2^k}\sum_{i=0}^{2^k-1}\ket i\ket i.
\end{equation}
A \emph{stabilizer group} $\sf S$ on the ($n-k$)-qubit subsystem is a group of $2^{n-k}$ commuting Pauli operators. 
Mathematically, $\sf S$ can be viewed as an ($n-k$)-dimensional subspace of $\mbb Z_2^{2(n-k)}$.
The \emph{stabilizer states} $\ket{\phi_e^{\sf S}}$ are the simultaneous eigenstates of all Pauli operators in $\sf S$, which can be expressed as
\begin{equation}
	\ketbra{\phi^{\sf S}_e}{\phi^{\sf S}_e} = \frac{1}{2^{n-k}}\sum_{s\in{\sf S}}(-1)^\expval{s,e}P_s,
\end{equation}
for $e\in {\sf S}^{\perp}\coleq {\mbb Z}_2^{2(n-k)}/{\sf S}$ known as the \emph{error syndrome}. 
Note that $\{\ket{\phi_e^{\sf S}}\}_e$ forms an orthonormal basis for the $(n-k)$-qubit subsystem.
The \emph{stabilizer covering} $\sf O$ is a set of stabilizer groups $\{{\sf S}_i\}_i$ such that every Pauli operator belongs to at least one ${\sf S}_i$~\cite{flammia2020efficient}.

\begin{figure}[htb]
\begin{algorithm}[H]
		\caption{$k$-qubit-ancilla-assisted Pauli channel estimation}
		\label{alg:main}
		\begin{algorithmic}[1]
			\Require
			(1)~$N$ copies of an $n$-qubit Pauli channel $\Lambda$. (2)~A stabilizer covering of $\sf{P}^{n-k}$ denoted as $\sf O$.
			\Ensure
			Estimates $\widehat{\bm\lambda}$ for the Pauli eigenvalues of $\Lambda$.
            \State $\widehat\lambda_a \coleq 0$,~$N_a\coleq 0$~\textbf{for all}~$a\in\mbb Z^{2n}_2$.
			\For{$\sf S \in \sf O$}
			\For{$i=1~\textbf{to}~\floor{N/|\sf O|}$}
			\State Input $\ket{\Psi_0}\otimes\ket{\phi^{\sf S}_0}$ to $\mathds 1\otimes\Lambda$.
			\State Measure in basis $\{\ket{\Psi_v}\otimes\ket{\phi^{\sf S}_e}\}$ with outcomes $v,e$.
			\For{$u\in\mbb Z^{2k}_{2},s\in{\sf S}$}
			\State$\widehat{\lambda}_{u\oplus s} \pe (-1)^{\expval{u,v}+\expval{s,e}}$,~$N_{u\oplus s}\pe 1$.
			\EndFor
			\EndFor
			\EndFor
			\State $\widehat\lambda_a \coleq \widehat\lambda_a / N_a$~\textbf{for all}~$a\in\mbb Z^{2n}_2$.
            \State\Return$\widehat{\bm\lambda}\coleq\{\widehat\lambda_a\}_a$.
            \Statex \textit{(Note: $\oplus$ stands for string concatenation.)}
		\end{algorithmic}
\end{algorithm}
\end{figure}

\begin{figure}[!htbp]
	\centering
	\includegraphics[width = 0.85\columnwidth]{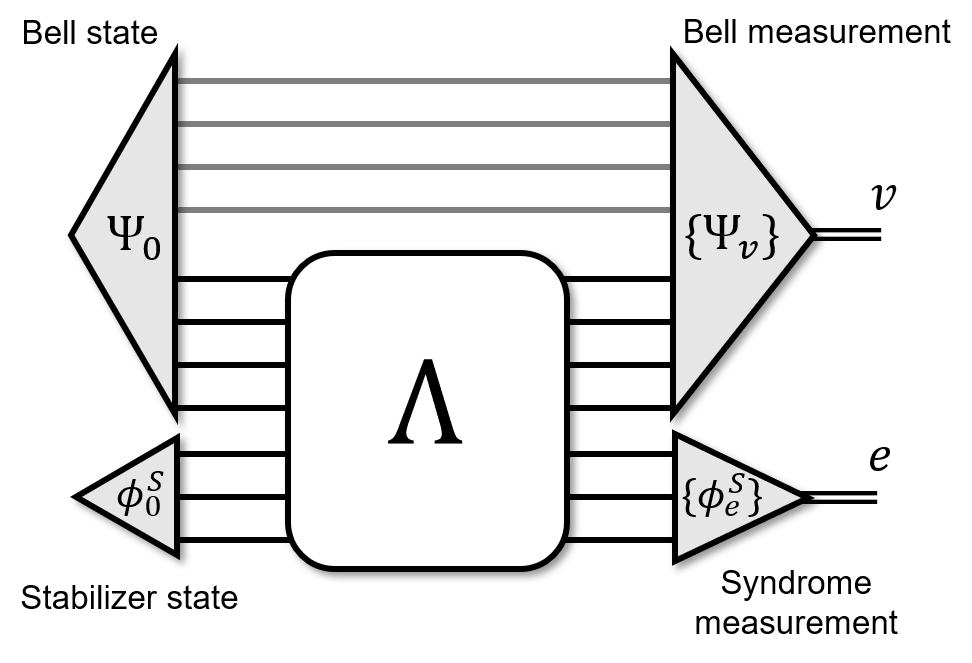}
	\caption{A single round of measurement for the $k$-qubit-assisted Pauli channel estimation protocol in Algorithm~\ref{alg:main}. 
	Here, a $4$-qubit ancilla is used to estimate a $7$-qubit Pauli channel.}
	\label{fig:alg}
\end{figure}

\begin{theorem}\label{th:upper}
Algorithm~\ref{alg:main} gives an estimate $\widehat{\bm{\lambda}}$ for the Pauli eigenvalues $\bm\lambda$ of any $n$-qubit Pauli channels that satisfies
\begin{equation}
    |\widehat{\lambda}_a-\lambda_a|\le\varepsilon,\quad \forall a\in \mbb Z_2^{2n}
\end{equation}
with success probability at least $1-\delta$, given the following number of samples
\begin{equation}
    N = \mathcal O (|{\sf O}|\times n\varepsilon^{-2}\log\delta^{-1}).
\end{equation}
\end{theorem}

\begin{proof}[Proof of Theorem~\ref{th:upper}]
The probability distribution of measurement outcomes at Line~5 in Algorithm~\ref{alg:main} can be calculated as
\begin{equation}\label{eq:PauliDistribution}
	p(v,e) = \frac{1}{2^{n+k}}\sum_{u\in \mbb Z^{2k}_2}\sum_{s\in \mbb {\sf S}}\lambda_{u\oplus s}(-1)^{\expval{u,v}}(-1)^{\expval{s,e}}.
\end{equation}
See the supplemental material for more details.
Therefore, $p(v,e)$ and $\lambda_{u\oplus s}$ are related by the Walsh-Hadamard transform. Taking the inverse transform, we get
\begin{equation}
	\lambda_{u\oplus s} = \sum_{v\in \mbb Z^{2k}_2}\sum_{e\in \mbb {\sf S}^\perp} p(v,e)(-1)^{\expval{u,v}+\expval{s,e}}.
\end{equation}
Thus, $(-1)^{\expval{u,v}+\expval{s,e}}$ is an un-biased estimator of $\lambda_{u\oplus s}$. According to Hoeffding's bound, $N_0 = \mc O(\varepsilon^{-2}\log\delta_0^{-1})$ samples are enough to estimate a single $\lambda_{u\oplus s}$ to additive precision $\varepsilon$ with success probability at least $1-\delta_0$. Since every $\lambda_a$ is covered by some stabilizer group ${\sf S}\in{\sf O}$, a total sample complexity of $N = \mathcal O (|{\sf O}|\times n\varepsilon^{-2}\log\delta^{-1})$ is enough to estimate all $\lambda_a$ to additive error $\varepsilon$ simultaneously with success probability at least $1-\delta$, by setting $\delta_0 \coleq 4^{-n}\delta$ and applying the union bound.
\end{proof}

\begin{corollary}\label{co:optimalupper}
	There exists a $k$-qubit-ancilla-assisted non-adaptive non-concatenating Pauli channel estimation protocol achieving $|\widehat{\lambda}_a-\lambda_a|\le\varepsilon$ for all $a\in\mbb Z_2^{2n}$ with probability $1-\delta$ using $N=\mc O(n2^{n-k}\varepsilon^{-2}\log\delta^{-1})$ samples.
\end{corollary}

\begin{proof}
	This follows from the existence of a stabilizer covering for ${\sf P}^{n-k}$ of size $2^{n-k}+1$~\cite{lawrence2002mutually}, which in turn follows from the existence of $2^{n-k}+1$ mutually-unbiased bases for $(n-k)$-qubit systems~\cite{wootters1989optimal}.  
\end{proof}

\noindent We remark that, the choice of stabilizer covering in Corollary~\ref{co:optimalupper} gives the optimal sample complexity for all \emph{non-adaptive} and \emph{non-concatenating} Pauli eigenvalues estimation protocols, as proved in the next section.
From a practical point of view, it involves $(n-k)$-qubit stabilizer states that might be difficult to prepare.
A more experimental friendly version is to choose the stabilizer covering generated by all $3^{n-k}$ possible Pauli measurements, in which case only Pauli eigenstates preparation and Pauli measurements are required (on the $(n-k)$-qubit subsystem), at the expense of a sub-optimal sample complexity $N=\mc O(n3^{n-k}\varepsilon^{-2}\log\delta^{-1})$.

\textit{Lower bounds.{\textemdash}} 
We have established a sample complexity upper bound of $\mc O(n2^{n-k})$ for non-adaptive non-concatenating k-qubit-ancilla-assisted Pauli eigenvalues estimation protocols, which implies a $\mc O(n)$ upper bound for the $n$-qubit ancilla case and a $\mc O(n2^n)$ upper bound for the ancilla-free case. In the following Theorem~\ref{th:lower_main}, we provide corresponding lower bounds to justify the exponential advantage provided by ancilla in this task.

\begin{theorem}\label{th:lower_main}
	For any estimation protocol that give an estimate $\widehat{\bm{\lambda}}$ for the Pauli eigenvalues $\bm\lambda$ of an arbitrary unknwon $n$-qubit Pauli channel $\Lambda$ such that
	\begin{equation}
		|\widehat{\lambda}_a-\lambda_a|\le 1/2,\quad \forall a\in \mbb Z_2^{2n}
	\end{equation}
	holds with high probability, the number of samples of $\Lambda$ must satisfies (recall Fig.~\ref{fig:model})
	\begin{enumerate}[label=(\Alph*)]
		\item $N=\Omega(n2^{n-k})$, for non-adaptive non-concatenating k-qubit-ancilla measurements. 
		\item $N=\Omega(2^{(n-k)/3})$, for adaptive non-concatenating k-qubit-ancilla measurements. 
		
		\item $N\ge N'=\Omega(2^{n/3})$, for adaptive concatenating ancilla-free measurements, where $N'$ stands for the number of measurement rounds.
		
		\item $N=\Omega(n)$, for fully entangled measurements. 
	\end{enumerate}
\end{theorem}
Indeed, Theorem~\ref{th:lower_main} and Corollary~\ref{co:optimalupper} establishes an exponential advantage of ancilla-assisted measurements over ancilla-free measurements even with channel concatenation (as in the RB-type Pauli channel estimation protocols in Ref.~\cite{flammia2020efficient}). 
Furthermore, for the non-concatenating cases, we see a roughly matching bounds for all ancilla size $0\le k\le n$, which can be interpreted as that a small number of ancilla ($k=o(n)$) would not help much in improving the sample efficiency. 
We also see from (D) that the sample complexity of Algorithm~\ref{alg:main} with $n$ ancilla qubits is optimal among all entangled strategies, thus we need not study protocols with more than $n$ ancilla qubits.

\begin{proof}[Sketch of the proof]
	Our proof generalizes the techniques of Huang~\textit{et al.}~\cite{huang2021information} in proving lower bounds for learning Pauli expectation values of quantum states.
	The key is to construct the following set of Pauli channels 
	\begin{equation}\label{eq:construction}
		\left\{\Lambda_{(a,s)}(\cdot) = \frac{1}{2^n}(I\Tr(\cdot)+sP_a\Tr(P_a(\cdot)))\right\}_{a,s},
	\end{equation}
	for $a\in\{1,\cdots,4^n-1\}$ and $s=\pm1$.
	An estimation protocol satisfying the assumption of Theorem~\ref{th:lower_main} is able to identify an arbitrary element of this set using $N$ copies of the channel. We can then use information-theoretical arguments to lower bound $N$ for (A); The bounds in (B) and (C) are proved by reducing the learning problem to a channel discrimination problem between the completely-deplorizing channel and the channels in Eq.~\eqref{eq:construction}; 
	To prove the bound in (D), we first use \emph{teleportation  stretching}~\cite{ pirandola2017fundamental,pirandola2019fundamental} to reduce any estimation protocols on $N$ copies of the Pauli channel into a POVM measurement on $N$ copies of their Choi states~\cite{choi1975completely,jamiolkowski1972linear}, and then apply the Holevo's theorem~\cite{holevo1973bounds}. See the supplemental material for a full proof of Theorem~\ref{th:lower_main}.
\end{proof}

\emph{Applications in quantum benchmarking.{\textemdash}}
The estimation protocol we described earlier is based on an ideal situation where we have direct access to the Pauli channel of interest and ignore the state preparation and measurement (SPAM) error.
In the supplemental material, we introduce a quantum benchmarking protocol for the Pauli gates, which can be viewed as an extension of Algorithm~\ref{alg:main}.
The protocol uses an $n$-qubit ancilla and (possibly imperfect) Bell state preparation and measurements. However, the application of the Pauli channel in Fig.~\ref{fig:alg} is now replaced by applying a sequence of random Pauli gates to be characterized, as shown in Fig.~\ref{fig:spam}. By repeating this procedure for different sequence lengths, we reduce the sensitivity to SPAM errors.
In addition to robustness against SPAM errors, this protocol provides an exponential advantage in sample complexity compared to any ancilla-free protocol, as long as the ancillary system has a long coherence time and is well isolated from the main system, among other reasonable noise assumptions.
These assumptions can potentially be satisfied by \textit{e.g.}, an ion trap platform~\cite{wright2019benchmarking,pino2021demonstration}. 

\begin{figure}[htp]
	\centering
	\includegraphics[width=0.8\columnwidth]{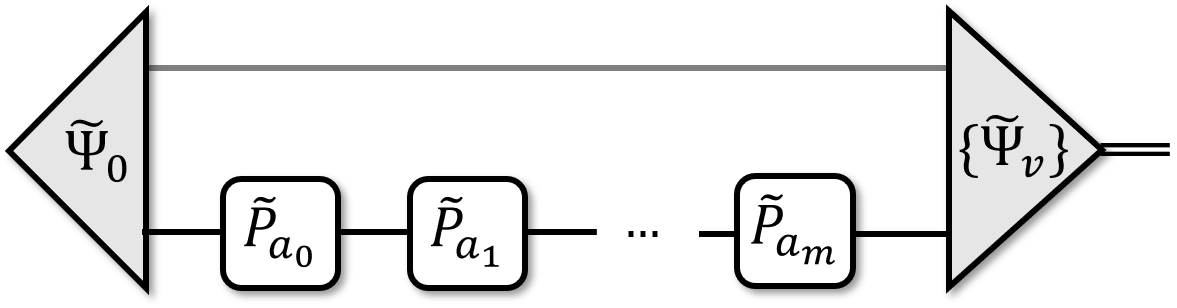}
	\caption{A single round of measurement for the $n$-qubit ancilla-assisted Pauli gate benchmarking protocol detailed in the supplemental material. Here $\widetilde{\Psi}_0$ and $\{ \widetilde{\Psi}_v \}$ stands for the (noisy) Bell states/measurements. $\{\widetilde{P}_{a_t}\}_{t=0}^m$ stands for a sequence of random (noisy) Pauli gates to be characterized.}
	\label{fig:spam}
\end{figure}

\textit{Summary and outlook.{\textemdash}}
In this work, we show a provable quantum advantage provided by entangled measurements for a learning task~\cite{aharonov2021quantum,huang2021information} of Pauli channel estimation, which is a practically useful tool urgently needed to characterize large quantum systems.
For quantum benchmarking, our results provide fundamental efficiency limits for Pauli noise estimation, which partly solve an open problem raised in~\cite{flammia2020efficient}.  
We also describe how the ancilla-assisted Pauli channel estimation protocol can be applied to a practical quantum benchmarking tasks in a noise-resilient and sample-efficient manner.
Our results provide a promising tool for both characterizing near-term quantum devices and demonstrating quantum advantages in those systems

Several interesting questions remain to be explored in the future, including exploring the quantum advantages in learning other properties of Pauli channels (e.g., the Pauli error rates~\footnote{
    It is shown in~\cite{flammia2021pauli} that the Pauli error rates for an $n$-qubit Pauli channel can be estimated to small error in $l_\infty$ distance using un-entangled measurements with $\mc O(\log n)$ samples, so there is no large advantages for entangled strategies. Still, one could expect a sample-efficiency advantage under other practically-interesting metric of estimation precision.
}); analyzing the sample complexity for learning Pauli channels with more specific structures; and analyzing the experimental performance of our algorithms in comparison to other ancilla-free protocols~\cite{harper2020efficient,hashim2020randomized}.

\begin{acknowledgments}
We acknowledge support from the ARO (W911NF-18-1-0020, W911NF-18-1-0212), ARO MURI (W911NF-16-1-0349), AFOSR MURI (FA9550-19-1-0399, FA9550-21-1-0209), DoE Q-NEXT, NSF (EFMA-1640959, OMA-1936118, EEC-1941583), NTT Research, and the Packard Foundation (2013-39273).
S.Z. acknowledges funding provided by the Institute for Quantum Information and Matter, an NSF Physics Frontiers Center (NSF Grant PHY-1733907).
A.S. is supported by a Chicago Prize Postdoctoral Fellowship in Theoretical Quantum Science.
\end{acknowledgments}

\bibliography{BibPauli}

\begin{thebibliography}{51}%
\makeatletter
\providecommand \@ifxundefined [1]{%
 \@ifx{#1\undefined}
}%
\providecommand \@ifnum [1]{%
 \ifnum #1\expandafter \@firstoftwo
 \else \expandafter \@secondoftwo
 \fi
}%
\providecommand \@ifx [1]{%
 \ifx #1\expandafter \@firstoftwo
 \else \expandafter \@secondoftwo
 \fi
}%
\providecommand \natexlab [1]{#1}%
\providecommand \enquote  [1]{``#1''}%
\providecommand \bibnamefont  [1]{#1}%
\providecommand \bibfnamefont [1]{#1}%
\providecommand \citenamefont [1]{#1}%
\providecommand \href@noop [0]{\@secondoftwo}%
\providecommand \href [0]{\begingroup \@sanitize@url \@href}%
\providecommand \@href[1]{\@@startlink{#1}\@@href}%
\providecommand \@@href[1]{\endgroup#1\@@endlink}%
\providecommand \@sanitize@url [0]{\catcode `\\12\catcode `\$12\catcode
  `\&12\catcode `\#12\catcode `\^12\catcode `\_12\catcode `\%12\relax}%
\providecommand \@@startlink[1]{}%
\providecommand \@@endlink[0]{}%
\providecommand \url  [0]{\begingroup\@sanitize@url \@url }%
\providecommand \@url [1]{\endgroup\@href {#1}{\urlprefix }}%
\providecommand \urlprefix  [0]{URL }%
\providecommand \Eprint [0]{\href }%
\providecommand \doibase [0]{https://doi.org/}%
\providecommand \selectlanguage [0]{\@gobble}%
\providecommand \bibinfo  [0]{\@secondoftwo}%
\providecommand \bibfield  [0]{\@secondoftwo}%
\providecommand \translation [1]{[#1]}%
\providecommand \BibitemOpen [0]{}%
\providecommand \bibitemStop [0]{}%
\providecommand \bibitemNoStop [0]{.\EOS\space}%
\providecommand \EOS [0]{\spacefactor3000\relax}%
\providecommand \BibitemShut  [1]{\csname bibitem#1\endcsname}%
\let\auto@bib@innerbib\@empty
\bibitem [{\citenamefont {Preskill}(2018)}]{preskill2018quantum}%
  \BibitemOpen
  \bibfield  {author} {\bibinfo {author} {\bibfnamefont {J.}~\bibnamefont
  {Preskill}},\ }\bibfield  {title} {\bibinfo {title} {Quantum computing in the
  nisq era and beyond},\ }\href {https://doi.org/10.22331/q-2018-08-06-79}
  {\bibfield  {journal} {\bibinfo  {journal} {Quantum}\ }\textbf {\bibinfo
  {volume} {2}},\ \bibinfo {pages} {79} (\bibinfo {year} {2018})}\BibitemShut
  {NoStop}%
\bibitem [{\citenamefont {Harrow}\ and\ \citenamefont
  {Montanaro}(2017)}]{harrow2017quantum}%
  \BibitemOpen
  \bibfield  {author} {\bibinfo {author} {\bibfnamefont {A.~W.}\ \bibnamefont
  {Harrow}}\ and\ \bibinfo {author} {\bibfnamefont {A.}~\bibnamefont
  {Montanaro}},\ }\bibfield  {title} {\bibinfo {title} {Quantum computational
  supremacy},\ }\href {https://doi.org/harrow2017quantum} {\bibfield  {journal}
  {\bibinfo  {journal} {Nature}\ }\textbf {\bibinfo {volume} {549}},\ \bibinfo
  {pages} {203} (\bibinfo {year} {2017})}\BibitemShut {NoStop}%
\bibitem [{\citenamefont {Arute}\ \emph {et~al.}(2019)\citenamefont {Arute},
  \citenamefont {Arya}, \citenamefont {Babbush}, \citenamefont {Bacon},
  \citenamefont {Bardin}, \citenamefont {Barends}, \citenamefont {Biswas},
  \citenamefont {Boixo}, \citenamefont {Brandao}, \citenamefont {Buell} \emph
  {et~al.}}]{arute2019quantum}%
  \BibitemOpen
  \bibfield  {author} {\bibinfo {author} {\bibfnamefont {F.}~\bibnamefont
  {Arute}}, \bibinfo {author} {\bibfnamefont {K.}~\bibnamefont {Arya}},
  \bibinfo {author} {\bibfnamefont {R.}~\bibnamefont {Babbush}}, \bibinfo
  {author} {\bibfnamefont {D.}~\bibnamefont {Bacon}}, \bibinfo {author}
  {\bibfnamefont {J.~C.}\ \bibnamefont {Bardin}}, \bibinfo {author}
  {\bibfnamefont {R.}~\bibnamefont {Barends}}, \bibinfo {author} {\bibfnamefont
  {R.}~\bibnamefont {Biswas}}, \bibinfo {author} {\bibfnamefont
  {S.}~\bibnamefont {Boixo}}, \bibinfo {author} {\bibfnamefont {F.~G.}\
  \bibnamefont {Brandao}}, \bibinfo {author} {\bibfnamefont {D.~A.}\
  \bibnamefont {Buell}}, \emph {et~al.},\ }\bibfield  {title} {\bibinfo {title}
  {Quantum supremacy using a programmable superconducting processor},\ }\href
  {https://doi.org/10.1038/s41586-019-1666-5} {\bibfield  {journal} {\bibinfo
  {journal} {Nature}\ }\textbf {\bibinfo {volume} {574}},\ \bibinfo {pages}
  {505} (\bibinfo {year} {2019})}\BibitemShut {NoStop}%
\bibitem [{\citenamefont {Zhong}\ \emph {et~al.}(2020)\citenamefont {Zhong},
  \citenamefont {Wang}, \citenamefont {Deng}, \citenamefont {Chen},
  \citenamefont {Peng}, \citenamefont {Luo}, \citenamefont {Qin}, \citenamefont
  {Wu}, \citenamefont {Ding}, \citenamefont {Hu} \emph
  {et~al.}}]{zhong2020quantum}%
  \BibitemOpen
  \bibfield  {author} {\bibinfo {author} {\bibfnamefont {H.-S.}\ \bibnamefont
  {Zhong}}, \bibinfo {author} {\bibfnamefont {H.}~\bibnamefont {Wang}},
  \bibinfo {author} {\bibfnamefont {Y.-H.}\ \bibnamefont {Deng}}, \bibinfo
  {author} {\bibfnamefont {M.-C.}\ \bibnamefont {Chen}}, \bibinfo {author}
  {\bibfnamefont {L.-C.}\ \bibnamefont {Peng}}, \bibinfo {author}
  {\bibfnamefont {Y.-H.}\ \bibnamefont {Luo}}, \bibinfo {author} {\bibfnamefont
  {J.}~\bibnamefont {Qin}}, \bibinfo {author} {\bibfnamefont {D.}~\bibnamefont
  {Wu}}, \bibinfo {author} {\bibfnamefont {X.}~\bibnamefont {Ding}}, \bibinfo
  {author} {\bibfnamefont {Y.}~\bibnamefont {Hu}}, \emph {et~al.},\ }\bibfield
  {title} {\bibinfo {title} {Quantum computational advantage using photons},\
  }\href {https://doi.org/10.1126/science.abe8770} {\bibfield  {journal}
  {\bibinfo  {journal} {Science}\ }\textbf {\bibinfo {volume} {370}},\ \bibinfo
  {pages} {1460} (\bibinfo {year} {2020})}\BibitemShut {NoStop}%
\bibitem [{\citenamefont {Wu}\ \emph {et~al.}(2021)\citenamefont {Wu},
  \citenamefont {Bao}, \citenamefont {Cao}, \citenamefont {Chen}, \citenamefont
  {Chen}, \citenamefont {Chen}, \citenamefont {Chung}, \citenamefont {Deng},
  \citenamefont {Du}, \citenamefont {Fan} \emph {et~al.}}]{wu2021strong}%
  \BibitemOpen
  \bibfield  {author} {\bibinfo {author} {\bibfnamefont {Y.}~\bibnamefont
  {Wu}}, \bibinfo {author} {\bibfnamefont {W.-S.}\ \bibnamefont {Bao}},
  \bibinfo {author} {\bibfnamefont {S.}~\bibnamefont {Cao}}, \bibinfo {author}
  {\bibfnamefont {F.}~\bibnamefont {Chen}}, \bibinfo {author} {\bibfnamefont
  {M.-C.}\ \bibnamefont {Chen}}, \bibinfo {author} {\bibfnamefont
  {X.}~\bibnamefont {Chen}}, \bibinfo {author} {\bibfnamefont {T.-H.}\
  \bibnamefont {Chung}}, \bibinfo {author} {\bibfnamefont {H.}~\bibnamefont
  {Deng}}, \bibinfo {author} {\bibfnamefont {Y.}~\bibnamefont {Du}}, \bibinfo
  {author} {\bibfnamefont {D.}~\bibnamefont {Fan}}, \emph {et~al.},\ }\bibfield
   {title} {\bibinfo {title} {Strong quantum computational advantage using a
  superconducting quantum processor},\ }\href@noop {} {\bibfield  {journal}
  {\bibinfo  {journal} {arXiv preprint arXiv:2106.14734}\ } (\bibinfo {year}
  {2021})}\BibitemShut {NoStop}%
\bibitem [{\citenamefont {Bubeck}\ \emph {et~al.}(2020)\citenamefont {Bubeck},
  \citenamefont {Chen},\ and\ \citenamefont {Li}}]{bubeck2020entanglement}%
  \BibitemOpen
  \bibfield  {author} {\bibinfo {author} {\bibfnamefont {S.}~\bibnamefont
  {Bubeck}}, \bibinfo {author} {\bibfnamefont {S.}~\bibnamefont {Chen}},\ and\
  \bibinfo {author} {\bibfnamefont {J.}~\bibnamefont {Li}},\ }\bibfield
  {title} {\bibinfo {title} {Entanglement is necessary for optimal quantum
  property testing},\ }in\ \href {https://doi.org/10.1109/FOCS46700.2020.00070}
  {\emph {\bibinfo {booktitle} {2020 IEEE 61st Annual Symposium on Foundations
  of Computer Science (FOCS)}}}\ (\bibinfo {organization} {IEEE},\ \bibinfo
  {year} {2020})\ pp.\ \bibinfo {pages} {692--703}\BibitemShut {NoStop}%
\bibitem [{\citenamefont {Aharonov}\ \emph {et~al.}(2021)\citenamefont
  {Aharonov}, \citenamefont {Cotler},\ and\ \citenamefont
  {Qi}}]{aharonov2021quantum}%
  \BibitemOpen
  \bibfield  {author} {\bibinfo {author} {\bibfnamefont {D.}~\bibnamefont
  {Aharonov}}, \bibinfo {author} {\bibfnamefont {J.}~\bibnamefont {Cotler}},\
  and\ \bibinfo {author} {\bibfnamefont {X.-L.}\ \bibnamefont {Qi}},\
  }\bibfield  {title} {\bibinfo {title} {Quantum algorithmic measurement},\
  }\href@noop {} {\bibfield  {journal} {\bibinfo  {journal} {arXiv preprint
  arXiv:2101.04634}\ } (\bibinfo {year} {2021})}\BibitemShut {NoStop}%
\bibitem [{\citenamefont {Huang}\ \emph {et~al.}(2021)\citenamefont {Huang},
  \citenamefont {Kueng},\ and\ \citenamefont
  {Preskill}}]{huang2021information}%
  \BibitemOpen
  \bibfield  {author} {\bibinfo {author} {\bibfnamefont {H.-Y.}\ \bibnamefont
  {Huang}}, \bibinfo {author} {\bibfnamefont {R.}~\bibnamefont {Kueng}},\ and\
  \bibinfo {author} {\bibfnamefont {J.}~\bibnamefont {Preskill}},\ }\bibfield
  {title} {\bibinfo {title} {Information-theoretic bounds on quantum advantage
  in machine learning},\ }\href
  {https://doi.org/10.1103/PhysRevLett.126.190505} {\bibfield  {journal}
  {\bibinfo  {journal} {Phys. Rev. Lett.}\ }\textbf {\bibinfo {volume} {126}},\
  \bibinfo {pages} {190505} (\bibinfo {year} {2021})}\BibitemShut {NoStop}%
\bibitem [{\citenamefont {Chen}\ \emph
  {et~al.}(2021{\natexlab{a}})\citenamefont {Chen}, \citenamefont {Li},\ and\
  \citenamefont {O'Donnell}}]{chen2021toward}%
  \BibitemOpen
  \bibfield  {author} {\bibinfo {author} {\bibfnamefont {S.}~\bibnamefont
  {Chen}}, \bibinfo {author} {\bibfnamefont {J.}~\bibnamefont {Li}},\ and\
  \bibinfo {author} {\bibfnamefont {R.}~\bibnamefont {O'Donnell}},\ }\bibfield
  {title} {\bibinfo {title} {Toward instance-optimal state certification with
  incoherent measurements},\ }\href@noop {} {\bibfield  {journal} {\bibinfo
  {journal} {arXiv preprint arXiv:2102.13098}\ } (\bibinfo {year}
  {2021}{\natexlab{a}})}\BibitemShut {NoStop}%
\bibitem [{\citenamefont {Rossi}\ \emph {et~al.}(2021)\citenamefont {Rossi},
  \citenamefont {Yu}, \citenamefont {Chuang},\ and\ \citenamefont
  {Sugiura}}]{rossi2021quantum}%
  \BibitemOpen
  \bibfield  {author} {\bibinfo {author} {\bibfnamefont {Z.~M.}\ \bibnamefont
  {Rossi}}, \bibinfo {author} {\bibfnamefont {J.}~\bibnamefont {Yu}}, \bibinfo
  {author} {\bibfnamefont {I.~L.}\ \bibnamefont {Chuang}},\ and\ \bibinfo
  {author} {\bibfnamefont {S.}~\bibnamefont {Sugiura}},\ }\bibfield  {title}
  {\bibinfo {title} {Quantum advantage for noisy channel discrimination},\
  }\href@noop {} {\bibfield  {journal} {\bibinfo  {journal} {arXiv preprint
  arXiv:2105.08707}\ } (\bibinfo {year} {2021})}\BibitemShut {NoStop}%
\bibitem [{\citenamefont {Eisert}\ \emph {et~al.}(2020)\citenamefont {Eisert},
  \citenamefont {Hangleiter}, \citenamefont {Walk}, \citenamefont {Roth},
  \citenamefont {Markham}, \citenamefont {Parekh}, \citenamefont {Chabaud},\
  and\ \citenamefont {Kashefi}}]{eisert2020quantum}%
  \BibitemOpen
  \bibfield  {author} {\bibinfo {author} {\bibfnamefont {J.}~\bibnamefont
  {Eisert}}, \bibinfo {author} {\bibfnamefont {D.}~\bibnamefont {Hangleiter}},
  \bibinfo {author} {\bibfnamefont {N.}~\bibnamefont {Walk}}, \bibinfo {author}
  {\bibfnamefont {I.}~\bibnamefont {Roth}}, \bibinfo {author} {\bibfnamefont
  {D.}~\bibnamefont {Markham}}, \bibinfo {author} {\bibfnamefont
  {R.}~\bibnamefont {Parekh}}, \bibinfo {author} {\bibfnamefont
  {U.}~\bibnamefont {Chabaud}},\ and\ \bibinfo {author} {\bibfnamefont
  {E.}~\bibnamefont {Kashefi}},\ }\bibfield  {title} {\bibinfo {title} {Quantum
  certification and benchmarking},\ }\href
  {https://doi.org/10.1038/s42254-020-0186-4} {\bibfield  {journal} {\bibinfo
  {journal} {Nature Reviews Physics}\ }\textbf {\bibinfo {volume} {2}},\
  \bibinfo {pages} {382} (\bibinfo {year} {2020})}\BibitemShut {NoStop}%
\bibitem [{\citenamefont {Wallman}\ and\ \citenamefont
  {Emerson}(2016)}]{wallman2016noise}%
  \BibitemOpen
  \bibfield  {author} {\bibinfo {author} {\bibfnamefont {J.~J.}\ \bibnamefont
  {Wallman}}\ and\ \bibinfo {author} {\bibfnamefont {J.}~\bibnamefont
  {Emerson}},\ }\bibfield  {title} {\bibinfo {title} {Noise tailoring for
  scalable quantum computation via randomized compiling},\ }\href
  {https://doi.org/10.1103/PhysRevA.94.052325} {\bibfield  {journal} {\bibinfo
  {journal} {Physical Review A}\ }\textbf {\bibinfo {volume} {94}},\ \bibinfo
  {pages} {052325} (\bibinfo {year} {2016})}\BibitemShut {NoStop}%
\bibitem [{\citenamefont {Hashim}\ \emph {et~al.}(2020)\citenamefont {Hashim},
  \citenamefont {Naik}, \citenamefont {Morvan}, \citenamefont {Ville},
  \citenamefont {Mitchell}, \citenamefont {Kreikebaum}, \citenamefont {Davis},
  \citenamefont {Smith}, \citenamefont {Iancu}, \citenamefont {O'Brien} \emph
  {et~al.}}]{hashim2020randomized}%
  \BibitemOpen
  \bibfield  {author} {\bibinfo {author} {\bibfnamefont {A.}~\bibnamefont
  {Hashim}}, \bibinfo {author} {\bibfnamefont {R.~K.}\ \bibnamefont {Naik}},
  \bibinfo {author} {\bibfnamefont {A.}~\bibnamefont {Morvan}}, \bibinfo
  {author} {\bibfnamefont {J.-L.}\ \bibnamefont {Ville}}, \bibinfo {author}
  {\bibfnamefont {B.}~\bibnamefont {Mitchell}}, \bibinfo {author}
  {\bibfnamefont {J.~M.}\ \bibnamefont {Kreikebaum}}, \bibinfo {author}
  {\bibfnamefont {M.}~\bibnamefont {Davis}}, \bibinfo {author} {\bibfnamefont
  {E.}~\bibnamefont {Smith}}, \bibinfo {author} {\bibfnamefont
  {C.}~\bibnamefont {Iancu}}, \bibinfo {author} {\bibfnamefont {K.~P.}\
  \bibnamefont {O'Brien}}, \emph {et~al.},\ }\bibfield  {title} {\bibinfo
  {title} {Randomized compiling for scalable quantum computing on a noisy
  superconducting quantum processor},\ }\href@noop {} {\bibfield  {journal}
  {\bibinfo  {journal} {arXiv preprint arXiv:2010.00215}\ } (\bibinfo {year}
  {2020})}\BibitemShut {NoStop}%
\bibitem [{\citenamefont {Erhard}\ \emph {et~al.}(2019)\citenamefont {Erhard},
  \citenamefont {Wallman}, \citenamefont {Postler}, \citenamefont {Meth},
  \citenamefont {Stricker}, \citenamefont {Martinez}, \citenamefont
  {Schindler}, \citenamefont {Monz}, \citenamefont {Emerson},\ and\
  \citenamefont {Blatt}}]{erhard2019characterizing}%
  \BibitemOpen
  \bibfield  {author} {\bibinfo {author} {\bibfnamefont {A.}~\bibnamefont
  {Erhard}}, \bibinfo {author} {\bibfnamefont {J.~J.}\ \bibnamefont {Wallman}},
  \bibinfo {author} {\bibfnamefont {L.}~\bibnamefont {Postler}}, \bibinfo
  {author} {\bibfnamefont {M.}~\bibnamefont {Meth}}, \bibinfo {author}
  {\bibfnamefont {R.}~\bibnamefont {Stricker}}, \bibinfo {author}
  {\bibfnamefont {E.~A.}\ \bibnamefont {Martinez}}, \bibinfo {author}
  {\bibfnamefont {P.}~\bibnamefont {Schindler}}, \bibinfo {author}
  {\bibfnamefont {T.}~\bibnamefont {Monz}}, \bibinfo {author} {\bibfnamefont
  {J.}~\bibnamefont {Emerson}},\ and\ \bibinfo {author} {\bibfnamefont
  {R.}~\bibnamefont {Blatt}},\ }\bibfield  {title} {\bibinfo {title}
  {Characterizing large-scale quantum computers via cycle benchmarking},\
  }\href {https://doi.org/10.1038/s41467-019-13068-7} {\bibfield  {journal}
  {\bibinfo  {journal} {Nature communications}\ }\textbf {\bibinfo {volume}
  {10}},\ \bibinfo {pages} {1} (\bibinfo {year} {2019})}\BibitemShut {NoStop}%
\bibitem [{\citenamefont {Magesan}\ \emph {et~al.}(2011)\citenamefont
  {Magesan}, \citenamefont {Gambetta},\ and\ \citenamefont
  {Emerson}}]{magesan2011scalable}%
  \BibitemOpen
  \bibfield  {author} {\bibinfo {author} {\bibfnamefont {E.}~\bibnamefont
  {Magesan}}, \bibinfo {author} {\bibfnamefont {J.~M.}\ \bibnamefont
  {Gambetta}},\ and\ \bibinfo {author} {\bibfnamefont {J.}~\bibnamefont
  {Emerson}},\ }\bibfield  {title} {\bibinfo {title} {Scalable and robust
  randomized benchmarking of quantum processes},\ }\href
  {https://doi.org/10.1103/PhysRevLett.106.180504} {\bibfield  {journal}
  {\bibinfo  {journal} {Physical review letters}\ }\textbf {\bibinfo {volume}
  {106}},\ \bibinfo {pages} {180504} (\bibinfo {year} {2011})}\BibitemShut
  {NoStop}%
\bibitem [{\citenamefont {Harper}\ \emph {et~al.}(2020)\citenamefont {Harper},
  \citenamefont {Flammia},\ and\ \citenamefont
  {Wallman}}]{harper2020efficient}%
  \BibitemOpen
  \bibfield  {author} {\bibinfo {author} {\bibfnamefont {R.}~\bibnamefont
  {Harper}}, \bibinfo {author} {\bibfnamefont {S.~T.}\ \bibnamefont
  {Flammia}},\ and\ \bibinfo {author} {\bibfnamefont {J.~J.}\ \bibnamefont
  {Wallman}},\ }\bibfield  {title} {\bibinfo {title} {Efficient learning of
  quantum noise},\ }\href {https://doi.org/10.1038/s41567-020-0992-8}
  {\bibfield  {journal} {\bibinfo  {journal} {Nature Physics}\ }\textbf
  {\bibinfo {volume} {16}},\ \bibinfo {pages} {1184} (\bibinfo {year}
  {2020})}\BibitemShut {NoStop}%
\bibitem [{\citenamefont {Liu}\ \emph {et~al.}(2021)\citenamefont {Liu},
  \citenamefont {Otten}, \citenamefont {Bassirianjahromi}, \citenamefont
  {Jiang},\ and\ \citenamefont {Fefferman}}]{liu2021benchmarking}%
  \BibitemOpen
  \bibfield  {author} {\bibinfo {author} {\bibfnamefont {Y.}~\bibnamefont
  {Liu}}, \bibinfo {author} {\bibfnamefont {M.}~\bibnamefont {Otten}}, \bibinfo
  {author} {\bibfnamefont {R.}~\bibnamefont {Bassirianjahromi}}, \bibinfo
  {author} {\bibfnamefont {L.}~\bibnamefont {Jiang}},\ and\ \bibinfo {author}
  {\bibfnamefont {B.}~\bibnamefont {Fefferman}},\ }\bibfield  {title} {\bibinfo
  {title} {Benchmarking near-term quantum computers via random circuit
  sampling},\ }\href@noop {} {\bibfield  {journal} {\bibinfo  {journal} {arXiv
  preprint arXiv:2105.05232}\ } (\bibinfo {year} {2021})}\BibitemShut {NoStop}%
\bibitem [{\citenamefont {Fujiwara}\ and\ \citenamefont
  {Imai}(2003)}]{fujiwara2003quantum}%
  \BibitemOpen
  \bibfield  {author} {\bibinfo {author} {\bibfnamefont {A.}~\bibnamefont
  {Fujiwara}}\ and\ \bibinfo {author} {\bibfnamefont {H.}~\bibnamefont
  {Imai}},\ }\bibfield  {title} {\bibinfo {title} {Quantum parameter estimation
  of a generalized pauli channel},\ }\href
  {https://doi.org/10.1088/0305-4470/36/29/314} {\bibfield  {journal} {\bibinfo
   {journal} {Journal of Physics A: Mathematical and General}\ }\textbf
  {\bibinfo {volume} {36}},\ \bibinfo {pages} {8093} (\bibinfo {year}
  {2003})}\BibitemShut {NoStop}%
\bibitem [{\citenamefont {Hayashi}(2010)}]{hayashi2010quantum}%
  \BibitemOpen
  \bibfield  {author} {\bibinfo {author} {\bibfnamefont {M.}~\bibnamefont
  {Hayashi}},\ }\bibfield  {title} {\bibinfo {title} {Quantum channel
  estimation and asymptotic bound},\ }in\ \href
  {https://doi.org/10.1088/1742-6596/233/1/012016} {\emph {\bibinfo {booktitle}
  {Journal of Physics: Conference Series}}},\ Vol.\ \bibinfo {volume} {233}\
  (\bibinfo {organization} {IOP Publishing},\ \bibinfo {year} {2010})\ p.\
  \bibinfo {pages} {012016}\BibitemShut {NoStop}%
\bibitem [{\citenamefont {Chiuri}\ \emph {et~al.}(2011)\citenamefont {Chiuri},
  \citenamefont {Rosati}, \citenamefont {Vallone}, \citenamefont {P{\'a}dua},
  \citenamefont {Imai}, \citenamefont {Giacomini}, \citenamefont
  {Macchiavello},\ and\ \citenamefont {Mataloni}}]{chiuri2011experimental}%
  \BibitemOpen
  \bibfield  {author} {\bibinfo {author} {\bibfnamefont {A.}~\bibnamefont
  {Chiuri}}, \bibinfo {author} {\bibfnamefont {V.}~\bibnamefont {Rosati}},
  \bibinfo {author} {\bibfnamefont {G.}~\bibnamefont {Vallone}}, \bibinfo
  {author} {\bibfnamefont {S.}~\bibnamefont {P{\'a}dua}}, \bibinfo {author}
  {\bibfnamefont {H.}~\bibnamefont {Imai}}, \bibinfo {author} {\bibfnamefont
  {S.}~\bibnamefont {Giacomini}}, \bibinfo {author} {\bibfnamefont
  {C.}~\bibnamefont {Macchiavello}},\ and\ \bibinfo {author} {\bibfnamefont
  {P.}~\bibnamefont {Mataloni}},\ }\bibfield  {title} {\bibinfo {title}
  {Experimental realization of optimal noise estimation for a general pauli
  channel},\ }\href {https://doi.org/10.1103/PhysRevLett.107.253602} {\bibfield
   {journal} {\bibinfo  {journal} {Physical review letters}\ }\textbf {\bibinfo
  {volume} {107}},\ \bibinfo {pages} {253602} (\bibinfo {year}
  {2011})}\BibitemShut {NoStop}%
\bibitem [{\citenamefont {Ruppert}\ \emph {et~al.}(2012)\citenamefont
  {Ruppert}, \citenamefont {Virosztek},\ and\ \citenamefont
  {Hangos}}]{ruppert2012optimal}%
  \BibitemOpen
  \bibfield  {author} {\bibinfo {author} {\bibfnamefont {L.}~\bibnamefont
  {Ruppert}}, \bibinfo {author} {\bibfnamefont {D.}~\bibnamefont {Virosztek}},\
  and\ \bibinfo {author} {\bibfnamefont {K.}~\bibnamefont {Hangos}},\
  }\bibfield  {title} {\bibinfo {title} {Optimal parameter estimation of pauli
  channels},\ }\href {https://doi.org/10.1088/1751-8113/45/26/265305}
  {\bibfield  {journal} {\bibinfo  {journal} {Journal of Physics A:
  Mathematical and Theoretical}\ }\textbf {\bibinfo {volume} {45}},\ \bibinfo
  {pages} {265305} (\bibinfo {year} {2012})}\BibitemShut {NoStop}%
\bibitem [{\citenamefont {Collins}(2013)}]{collins2013mixed}%
  \BibitemOpen
  \bibfield  {author} {\bibinfo {author} {\bibfnamefont {D.}~\bibnamefont
  {Collins}},\ }\bibfield  {title} {\bibinfo {title} {Mixed-state pauli-channel
  parameter estimation},\ }\href {https://doi.org/10.1103/PhysRevA.87.032301}
  {\bibfield  {journal} {\bibinfo  {journal} {Physical Review A}\ }\textbf
  {\bibinfo {volume} {87}},\ \bibinfo {pages} {032301} (\bibinfo {year}
  {2013})}\BibitemShut {NoStop}%
\bibitem [{\citenamefont {Flammia}\ and\ \citenamefont
  {Wallman}(2020)}]{flammia2020efficient}%
  \BibitemOpen
  \bibfield  {author} {\bibinfo {author} {\bibfnamefont {S.~T.}\ \bibnamefont
  {Flammia}}\ and\ \bibinfo {author} {\bibfnamefont {J.~J.}\ \bibnamefont
  {Wallman}},\ }\bibfield  {title} {\bibinfo {title} {Efficient estimation of
  pauli channels},\ }\href {https://doi.org/10.1145/3408039} {\bibfield
  {journal} {\bibinfo  {journal} {ACM Transactions on Quantum Computing}\
  }\textbf {\bibinfo {volume} {1}},\ \bibinfo {pages} {1} (\bibinfo {year}
  {2020})}\BibitemShut {NoStop}%
\bibitem [{\citenamefont {Harper}\ \emph {et~al.}(2021)\citenamefont {Harper},
  \citenamefont {Yu},\ and\ \citenamefont {Flammia}}]{harper2021fast}%
  \BibitemOpen
  \bibfield  {author} {\bibinfo {author} {\bibfnamefont {R.}~\bibnamefont
  {Harper}}, \bibinfo {author} {\bibfnamefont {W.}~\bibnamefont {Yu}},\ and\
  \bibinfo {author} {\bibfnamefont {S.~T.}\ \bibnamefont {Flammia}},\
  }\bibfield  {title} {\bibinfo {title} {Fast estimation of sparse quantum
  noise},\ }\href {https://doi.org/10.1103/PRXQuantum.2.010322} {\bibfield
  {journal} {\bibinfo  {journal} {PRX Quantum}\ }\textbf {\bibinfo {volume}
  {2}},\ \bibinfo {pages} {010322} (\bibinfo {year} {2021})}\BibitemShut
  {NoStop}%
\bibitem [{\citenamefont {Flammia}\ and\ \citenamefont
  {O'Donnell}(2021)}]{flammia2021pauli}%
  \BibitemOpen
  \bibfield  {author} {\bibinfo {author} {\bibfnamefont {S.~T.}\ \bibnamefont
  {Flammia}}\ and\ \bibinfo {author} {\bibfnamefont {R.}~\bibnamefont
  {O'Donnell}},\ }\bibfield  {title} {\bibinfo {title} {Pauli error estimation
  via population recovery},\ }\href {https://doi.org/10.22331/q-2021-09-23-549}
  {\bibfield  {journal} {\bibinfo  {journal} {Quantum}\ }\textbf {\bibinfo
  {volume} {5}},\ \bibinfo {pages} {549} (\bibinfo {year} {2021})}\BibitemShut
  {NoStop}%
\bibitem [{\citenamefont {Mohseni}\ \emph {et~al.}(2008)\citenamefont
  {Mohseni}, \citenamefont {Rezakhani},\ and\ \citenamefont
  {Lidar}}]{mohseni2008quantum}%
  \BibitemOpen
  \bibfield  {author} {\bibinfo {author} {\bibfnamefont {M.}~\bibnamefont
  {Mohseni}}, \bibinfo {author} {\bibfnamefont {A.~T.}\ \bibnamefont
  {Rezakhani}},\ and\ \bibinfo {author} {\bibfnamefont {D.~A.}\ \bibnamefont
  {Lidar}},\ }\bibfield  {title} {\bibinfo {title} {Quantum-process tomography:
  Resource analysis of different strategies},\ }\href
  {https://doi.org/10.1103/PhysRevA.77.032322} {\bibfield  {journal} {\bibinfo
  {journal} {Physical Review A}\ }\textbf {\bibinfo {volume} {77}},\ \bibinfo
  {pages} {032322} (\bibinfo {year} {2008})}\BibitemShut {NoStop}%
\bibitem [{\citenamefont {Knill}\ \emph {et~al.}(2008)\citenamefont {Knill},
  \citenamefont {Leibfried}, \citenamefont {Reichle}, \citenamefont {Britton},
  \citenamefont {Blakestad}, \citenamefont {Jost}, \citenamefont {Langer},
  \citenamefont {Ozeri}, \citenamefont {Seidelin},\ and\ \citenamefont
  {Wineland}}]{knill2008randomized}%
  \BibitemOpen
  \bibfield  {author} {\bibinfo {author} {\bibfnamefont {E.}~\bibnamefont
  {Knill}}, \bibinfo {author} {\bibfnamefont {D.}~\bibnamefont {Leibfried}},
  \bibinfo {author} {\bibfnamefont {R.}~\bibnamefont {Reichle}}, \bibinfo
  {author} {\bibfnamefont {J.}~\bibnamefont {Britton}}, \bibinfo {author}
  {\bibfnamefont {R.~B.}\ \bibnamefont {Blakestad}}, \bibinfo {author}
  {\bibfnamefont {J.~D.}\ \bibnamefont {Jost}}, \bibinfo {author}
  {\bibfnamefont {C.}~\bibnamefont {Langer}}, \bibinfo {author} {\bibfnamefont
  {R.}~\bibnamefont {Ozeri}}, \bibinfo {author} {\bibfnamefont
  {S.}~\bibnamefont {Seidelin}},\ and\ \bibinfo {author} {\bibfnamefont
  {D.~J.}\ \bibnamefont {Wineland}},\ }\bibfield  {title} {\bibinfo {title}
  {Randomized benchmarking of quantum gates},\ }\href
  {https://doi.org/10.1103/PhysRevA.77.012307} {\bibfield  {journal} {\bibinfo
  {journal} {Physical Review A}\ }\textbf {\bibinfo {volume} {77}},\ \bibinfo
  {pages} {012307} (\bibinfo {year} {2008})}\BibitemShut {NoStop}%
\bibitem [{\citenamefont {Nielsen}\ and\ \citenamefont
  {Chuang}(2011)}]{nielsen2002quantum}%
  \BibitemOpen
  \bibfield  {author} {\bibinfo {author} {\bibfnamefont {M.~A.}\ \bibnamefont
  {Nielsen}}\ and\ \bibinfo {author} {\bibfnamefont {I.~L.}\ \bibnamefont
  {Chuang}},\ }\href@noop {} {\emph {\bibinfo {title} {Quantum Computation and
  Quantum Information: 10th Anniversary Edition}}},\ \bibinfo {edition} {10th}\
  ed.\ (\bibinfo  {publisher} {Cambridge University Press},\ \bibinfo {address}
  {USA},\ \bibinfo {year} {2011})\BibitemShut {NoStop}%
\bibitem [{\citenamefont {Shor}(1996)}]{shor1996fault}%
  \BibitemOpen
  \bibfield  {author} {\bibinfo {author} {\bibfnamefont {P.~W.}\ \bibnamefont
  {Shor}},\ }\bibfield  {title} {\bibinfo {title} {Fault-tolerant quantum
  computation},\ }in\ \href {https://doi.org/10.1109/SFCS.1996.548464} {\emph
  {\bibinfo {booktitle} {Proceedings of 37th Conference on Foundations of
  Computer Science}}}\ (\bibinfo {organization} {IEEE},\ \bibinfo {year}
  {1996})\ pp.\ \bibinfo {pages} {56--65}\BibitemShut {NoStop}%
\bibitem [{\citenamefont {Aharonov}\ and\ \citenamefont
  {Ben-Or}(2008)}]{aharonov2008fault}%
  \BibitemOpen
  \bibfield  {author} {\bibinfo {author} {\bibfnamefont {D.}~\bibnamefont
  {Aharonov}}\ and\ \bibinfo {author} {\bibfnamefont {M.}~\bibnamefont
  {Ben-Or}},\ }\bibfield  {title} {\bibinfo {title} {Fault-tolerant quantum
  computation with constant error rate},\ }\bibfield  {journal} {\bibinfo
  {journal} {SIAM Journal on Computing}\ }\href
  {https://doi.org/10.1137/S0097539799359385} {10.1137/S0097539799359385}
  (\bibinfo {year} {2008})\BibitemShut {NoStop}%
\bibitem [{\citenamefont {Chen}\ \emph
  {et~al.}(2021{\natexlab{b}})\citenamefont {Chen}, \citenamefont {Yu},
  \citenamefont {Zeng},\ and\ \citenamefont {Flammia}}]{chen2020robust}%
  \BibitemOpen
  \bibfield  {author} {\bibinfo {author} {\bibfnamefont {S.}~\bibnamefont
  {Chen}}, \bibinfo {author} {\bibfnamefont {W.}~\bibnamefont {Yu}}, \bibinfo
  {author} {\bibfnamefont {P.}~\bibnamefont {Zeng}},\ and\ \bibinfo {author}
  {\bibfnamefont {S.~T.}\ \bibnamefont {Flammia}},\ }\bibfield  {title}
  {\bibinfo {title} {Robust shadow estimation},\ }\href
  {https://doi.org/10.1103/PRXQuantum.2.030348} {\bibfield  {journal} {\bibinfo
   {journal} {PRX Quantum}\ }\textbf {\bibinfo {volume} {2}},\ \bibinfo {pages}
  {030348} (\bibinfo {year} {2021}{\natexlab{b}})}\BibitemShut {NoStop}%
\bibitem [{\citenamefont {Helsen}\ \emph {et~al.}(2020)\citenamefont {Helsen},
  \citenamefont {Roth}, \citenamefont {Onorati}, \citenamefont {Werner},\ and\
  \citenamefont {Eisert}}]{helsen2020general}%
  \BibitemOpen
  \bibfield  {author} {\bibinfo {author} {\bibfnamefont {J.}~\bibnamefont
  {Helsen}}, \bibinfo {author} {\bibfnamefont {I.}~\bibnamefont {Roth}},
  \bibinfo {author} {\bibfnamefont {E.}~\bibnamefont {Onorati}}, \bibinfo
  {author} {\bibfnamefont {A.~H.}\ \bibnamefont {Werner}},\ and\ \bibinfo
  {author} {\bibfnamefont {J.}~\bibnamefont {Eisert}},\ }\bibfield  {title}
  {\bibinfo {title} {A general framework for randomized benchmarking},\
  }\href@noop {} {\bibfield  {journal} {\bibinfo  {journal} {arXiv preprint
  arXiv:2010.07974}\ } (\bibinfo {year} {2020})}\BibitemShut {NoStop}%
\bibitem [{\citenamefont {Chitambar}\ and\ \citenamefont
  {Gour}(2019)}]{chitambar2019quantum}%
  \BibitemOpen
  \bibfield  {author} {\bibinfo {author} {\bibfnamefont {E.}~\bibnamefont
  {Chitambar}}\ and\ \bibinfo {author} {\bibfnamefont {G.}~\bibnamefont
  {Gour}},\ }\bibfield  {title} {\bibinfo {title} {Quantum resource theories},\
  }\href {https://doi.org/10.1103/RevModPhys.91.025001} {\bibfield  {journal}
  {\bibinfo  {journal} {Reviews of Modern Physics}\ }\textbf {\bibinfo {volume}
  {91}},\ \bibinfo {pages} {025001} (\bibinfo {year} {2019})}\BibitemShut
  {NoStop}%
\bibitem [{Note1()}]{Note1}%
  \BibitemOpen
  \bibinfo {note} {This can be viewed as a consequence of the well-known
  \protect \emph {superdense coding} protocol~\cite
  {bennett1992communication}.}\BibitemShut {Stop}%
\bibitem [{\citenamefont {Lawrence}\ \emph {et~al.}(2002)\citenamefont
  {Lawrence}, \citenamefont {Brukner},\ and\ \citenamefont
  {Zeilinger}}]{lawrence2002mutually}%
  \BibitemOpen
  \bibfield  {author} {\bibinfo {author} {\bibfnamefont {J.}~\bibnamefont
  {Lawrence}}, \bibinfo {author} {\bibfnamefont {{\v{C}}.}~\bibnamefont
  {Brukner}},\ and\ \bibinfo {author} {\bibfnamefont {A.}~\bibnamefont
  {Zeilinger}},\ }\bibfield  {title} {\bibinfo {title} {Mutually unbiased
  binary observable sets on n qubits},\ }\href
  {https://doi.org/10.1103/PhysRevA.65.032320} {\bibfield  {journal} {\bibinfo
  {journal} {Physical Review A}\ }\textbf {\bibinfo {volume} {65}},\ \bibinfo
  {pages} {032320} (\bibinfo {year} {2002})}\BibitemShut {NoStop}%
\bibitem [{\citenamefont {Wootters}\ and\ \citenamefont
  {Fields}(1989)}]{wootters1989optimal}%
  \BibitemOpen
  \bibfield  {author} {\bibinfo {author} {\bibfnamefont {W.~K.}\ \bibnamefont
  {Wootters}}\ and\ \bibinfo {author} {\bibfnamefont {B.~D.}\ \bibnamefont
  {Fields}},\ }\bibfield  {title} {\bibinfo {title} {Optimal
  state-determination by mutually unbiased measurements},\ }\href
  {https://doi.org/10.1016/0003-4916(89)90322-9} {\bibfield  {journal}
  {\bibinfo  {journal} {Annals of Physics}\ }\textbf {\bibinfo {volume}
  {191}},\ \bibinfo {pages} {363} (\bibinfo {year} {1989})}\BibitemShut
  {NoStop}%
\bibitem [{\citenamefont {Pirandola}\ \emph {et~al.}(2017)\citenamefont
  {Pirandola}, \citenamefont {Laurenza}, \citenamefont {Ottaviani},\ and\
  \citenamefont {Banchi}}]{pirandola2017fundamental}%
  \BibitemOpen
  \bibfield  {author} {\bibinfo {author} {\bibfnamefont {S.}~\bibnamefont
  {Pirandola}}, \bibinfo {author} {\bibfnamefont {R.}~\bibnamefont {Laurenza}},
  \bibinfo {author} {\bibfnamefont {C.}~\bibnamefont {Ottaviani}},\ and\
  \bibinfo {author} {\bibfnamefont {L.}~\bibnamefont {Banchi}},\ }\bibfield
  {title} {\bibinfo {title} {Fundamental limits of repeaterless quantum
  communications},\ }\href {https://doi.org/10.1038/ncomms15043} {\bibfield
  {journal} {\bibinfo  {journal} {Nature communications}\ }\textbf {\bibinfo
  {volume} {8}},\ \bibinfo {pages} {1} (\bibinfo {year} {2017})}\BibitemShut
  {NoStop}%
\bibitem [{\citenamefont {Pirandola}\ \emph {et~al.}(2019)\citenamefont
  {Pirandola}, \citenamefont {Laurenza}, \citenamefont {Lupo},\ and\
  \citenamefont {Pereira}}]{pirandola2019fundamental}%
  \BibitemOpen
  \bibfield  {author} {\bibinfo {author} {\bibfnamefont {S.}~\bibnamefont
  {Pirandola}}, \bibinfo {author} {\bibfnamefont {R.}~\bibnamefont {Laurenza}},
  \bibinfo {author} {\bibfnamefont {C.}~\bibnamefont {Lupo}},\ and\ \bibinfo
  {author} {\bibfnamefont {J.~L.}\ \bibnamefont {Pereira}},\ }\bibfield
  {title} {\bibinfo {title} {Fundamental limits to quantum channel
  discrimination},\ }\href {https://doi.org/10.1038/s41534-019-0162-y}
  {\bibfield  {journal} {\bibinfo  {journal} {npj Quantum Information}\
  }\textbf {\bibinfo {volume} {5}},\ \bibinfo {pages} {1} (\bibinfo {year}
  {2019})}\BibitemShut {NoStop}%
\bibitem [{\citenamefont {Choi}(1975)}]{choi1975completely}%
  \BibitemOpen
  \bibfield  {author} {\bibinfo {author} {\bibfnamefont {M.-D.}\ \bibnamefont
  {Choi}},\ }\bibfield  {title} {\bibinfo {title} {Completely positive linear
  maps on complex matrices},\ }\href
  {https://doi.org/10.1016/0024-3795(75)90075-0} {\bibfield  {journal}
  {\bibinfo  {journal} {Linear algebra and its applications}\ }\textbf
  {\bibinfo {volume} {10}},\ \bibinfo {pages} {285} (\bibinfo {year}
  {1975})}\BibitemShut {NoStop}%
\bibitem [{\citenamefont {Jamio{\l}kowski}(1972)}]{jamiolkowski1972linear}%
  \BibitemOpen
  \bibfield  {author} {\bibinfo {author} {\bibfnamefont {A.}~\bibnamefont
  {Jamio{\l}kowski}},\ }\bibfield  {title} {\bibinfo {title} {Linear
  transformations which preserve trace and positive semidefiniteness of
  operators},\ }\href {https://doi.org/10.1016/0034-4877(72)90011-0} {\bibfield
   {journal} {\bibinfo  {journal} {Reports on Mathematical Physics}\ }\textbf
  {\bibinfo {volume} {3}},\ \bibinfo {pages} {275} (\bibinfo {year}
  {1972})}\BibitemShut {NoStop}%
\bibitem [{\citenamefont {Holevo}(1973)}]{holevo1973bounds}%
  \BibitemOpen
  \bibfield  {author} {\bibinfo {author} {\bibfnamefont {A.~S.}\ \bibnamefont
  {Holevo}},\ }\bibfield  {title} {\bibinfo {title} {Bounds for the quantity of
  information transmitted by a quantum communication channel},\ }\href@noop {}
  {\bibfield  {journal} {\bibinfo  {journal} {Problemy Peredachi Informatsii}\
  }\textbf {\bibinfo {volume} {9}},\ \bibinfo {pages} {3} (\bibinfo {year}
  {1973})}\BibitemShut {NoStop}%
\bibitem [{\citenamefont {Wright}\ \emph {et~al.}(2019)\citenamefont {Wright},
  \citenamefont {Beck}, \citenamefont {Debnath}, \citenamefont {Amini},
  \citenamefont {Nam}, \citenamefont {Grzesiak}, \citenamefont {Chen},
  \citenamefont {Pisenti}, \citenamefont {Chmielewski}, \citenamefont {Collins}
  \emph {et~al.}}]{wright2019benchmarking}%
  \BibitemOpen
  \bibfield  {author} {\bibinfo {author} {\bibfnamefont {K.}~\bibnamefont
  {Wright}}, \bibinfo {author} {\bibfnamefont {K.}~\bibnamefont {Beck}},
  \bibinfo {author} {\bibfnamefont {S.}~\bibnamefont {Debnath}}, \bibinfo
  {author} {\bibfnamefont {J.}~\bibnamefont {Amini}}, \bibinfo {author}
  {\bibfnamefont {Y.}~\bibnamefont {Nam}}, \bibinfo {author} {\bibfnamefont
  {N.}~\bibnamefont {Grzesiak}}, \bibinfo {author} {\bibfnamefont {J.-S.}\
  \bibnamefont {Chen}}, \bibinfo {author} {\bibfnamefont {N.}~\bibnamefont
  {Pisenti}}, \bibinfo {author} {\bibfnamefont {M.}~\bibnamefont
  {Chmielewski}}, \bibinfo {author} {\bibfnamefont {C.}~\bibnamefont
  {Collins}}, \emph {et~al.},\ }\bibfield  {title} {\bibinfo {title}
  {Benchmarking an 11-qubit quantum computer},\ }\href
  {https://doi.org/10.1038/s41467-019-13534-2} {\bibfield  {journal} {\bibinfo
  {journal} {Nature communications}\ }\textbf {\bibinfo {volume} {10}},\
  \bibinfo {pages} {1} (\bibinfo {year} {2019})}\BibitemShut {NoStop}%
\bibitem [{\citenamefont {Pino}\ \emph {et~al.}(2021)\citenamefont {Pino},
  \citenamefont {Dreiling}, \citenamefont {Figgatt}, \citenamefont {Gaebler},
  \citenamefont {Moses}, \citenamefont {Allman}, \citenamefont {Baldwin},
  \citenamefont {Foss-Feig}, \citenamefont {Hayes}, \citenamefont {Mayer} \emph
  {et~al.}}]{pino2021demonstration}%
  \BibitemOpen
  \bibfield  {author} {\bibinfo {author} {\bibfnamefont {J.~M.}\ \bibnamefont
  {Pino}}, \bibinfo {author} {\bibfnamefont {J.~M.}\ \bibnamefont {Dreiling}},
  \bibinfo {author} {\bibfnamefont {C.}~\bibnamefont {Figgatt}}, \bibinfo
  {author} {\bibfnamefont {J.~P.}\ \bibnamefont {Gaebler}}, \bibinfo {author}
  {\bibfnamefont {S.~A.}\ \bibnamefont {Moses}}, \bibinfo {author}
  {\bibfnamefont {M.}~\bibnamefont {Allman}}, \bibinfo {author} {\bibfnamefont
  {C.}~\bibnamefont {Baldwin}}, \bibinfo {author} {\bibfnamefont
  {M.}~\bibnamefont {Foss-Feig}}, \bibinfo {author} {\bibfnamefont
  {D.}~\bibnamefont {Hayes}}, \bibinfo {author} {\bibfnamefont
  {K.}~\bibnamefont {Mayer}}, \emph {et~al.},\ }\bibfield  {title} {\bibinfo
  {title} {Demonstration of the trapped-ion quantum ccd computer
  architecture},\ }\href {https://doi.org/10.1038/s41586-021-03318-4}
  {\bibfield  {journal} {\bibinfo  {journal} {Nature}\ }\textbf {\bibinfo
  {volume} {592}},\ \bibinfo {pages} {209} (\bibinfo {year}
  {2021})}\BibitemShut {NoStop}%
\bibitem [{Note2()}]{Note2}%
  \BibitemOpen
  \bibinfo {note} {It is shown in~\cite {flammia2021pauli} that the Pauli error
  rates for an $n$-qubit Pauli channel can be estimated to small error in
  $l_\infty $ distance using un-entangled measurements with $\protect \mathcal
  O(\protect \qopname \relax o{log}n)$ samples, so there is no large advantages
  for entangled strategies. Still, one could expect a sample-efficiency
  advantage under other practically-interesting metric of estimation
  precision.}\BibitemShut {Stop}%
\bibitem [{\citenamefont {Bennett}\ and\ \citenamefont
  {Wiesner}(1992)}]{bennett1992communication}%
  \BibitemOpen
  \bibfield  {author} {\bibinfo {author} {\bibfnamefont {C.~H.}\ \bibnamefont
  {Bennett}}\ and\ \bibinfo {author} {\bibfnamefont {S.~J.}\ \bibnamefont
  {Wiesner}},\ }\bibfield  {title} {\bibinfo {title} {Communication via one-and
  two-particle operators on einstein-podolsky-rosen states},\ }\href
  {https://doi.org/10.1103/PhysRevLett.69.2881} {\bibfield  {journal} {\bibinfo
   {journal} {Physical review letters}\ }\textbf {\bibinfo {volume} {69}},\
  \bibinfo {pages} {2881} (\bibinfo {year} {1992})}\BibitemShut {NoStop}%
\bibitem [{\citenamefont {Flammia}\ \emph {et~al.}(2012)\citenamefont
  {Flammia}, \citenamefont {Gross}, \citenamefont {Liu},\ and\ \citenamefont
  {Eisert}}]{flammia2012quantum}%
  \BibitemOpen
  \bibfield  {author} {\bibinfo {author} {\bibfnamefont {S.~T.}\ \bibnamefont
  {Flammia}}, \bibinfo {author} {\bibfnamefont {D.}~\bibnamefont {Gross}},
  \bibinfo {author} {\bibfnamefont {Y.-K.}\ \bibnamefont {Liu}},\ and\ \bibinfo
  {author} {\bibfnamefont {J.}~\bibnamefont {Eisert}},\ }\bibfield  {title}
  {\bibinfo {title} {Quantum tomography via compressed sensing: error bounds,
  sample complexity and efficient estimators},\ }\href
  {https://doi.org/10.1088/1367-2630/14/9/095022} {\bibfield  {journal}
  {\bibinfo  {journal} {New Journal of Physics}\ }\textbf {\bibinfo {volume}
  {14}},\ \bibinfo {pages} {095022} (\bibinfo {year} {2012})}\BibitemShut
  {NoStop}%
\bibitem [{\citenamefont {Haah}\ \emph {et~al.}(2017)\citenamefont {Haah},
  \citenamefont {Harrow}, \citenamefont {Ji}, \citenamefont {Wu},\ and\
  \citenamefont {Yu}}]{haah2017sample}%
  \BibitemOpen
  \bibfield  {author} {\bibinfo {author} {\bibfnamefont {J.}~\bibnamefont
  {Haah}}, \bibinfo {author} {\bibfnamefont {A.~W.}\ \bibnamefont {Harrow}},
  \bibinfo {author} {\bibfnamefont {Z.}~\bibnamefont {Ji}}, \bibinfo {author}
  {\bibfnamefont {X.}~\bibnamefont {Wu}},\ and\ \bibinfo {author}
  {\bibfnamefont {N.}~\bibnamefont {Yu}},\ }\bibfield  {title} {\bibinfo
  {title} {Sample-optimal tomography of quantum states},\ }\href
  {https://doi.org/10.1109/TIT.2017.2719044} {\bibfield  {journal} {\bibinfo
  {journal} {IEEE Transactions on Information Theory}\ }\textbf {\bibinfo
  {volume} {63}},\ \bibinfo {pages} {5628} (\bibinfo {year}
  {2017})}\BibitemShut {NoStop}%
\bibitem [{\citenamefont {Roth}\ \emph {et~al.}(2018)\citenamefont {Roth},
  \citenamefont {Kueng}, \citenamefont {Kimmel}, \citenamefont {Liu},
  \citenamefont {Gross}, \citenamefont {Eisert},\ and\ \citenamefont
  {Kliesch}}]{roth2018recovering}%
  \BibitemOpen
  \bibfield  {author} {\bibinfo {author} {\bibfnamefont {I.}~\bibnamefont
  {Roth}}, \bibinfo {author} {\bibfnamefont {R.}~\bibnamefont {Kueng}},
  \bibinfo {author} {\bibfnamefont {S.}~\bibnamefont {Kimmel}}, \bibinfo
  {author} {\bibfnamefont {Y.-K.}\ \bibnamefont {Liu}}, \bibinfo {author}
  {\bibfnamefont {D.}~\bibnamefont {Gross}}, \bibinfo {author} {\bibfnamefont
  {J.}~\bibnamefont {Eisert}},\ and\ \bibinfo {author} {\bibfnamefont
  {M.}~\bibnamefont {Kliesch}},\ }\bibfield  {title} {\bibinfo {title}
  {Recovering quantum gates from few average gate fidelities},\ }\href
  {https://doi.org/10.1103/PhysRevLett.121.170502} {\bibfield  {journal}
  {\bibinfo  {journal} {Physical review letters}\ }\textbf {\bibinfo {volume}
  {121}},\ \bibinfo {pages} {170502} (\bibinfo {year} {2018})}\BibitemShut
  {NoStop}%
\bibitem [{\citenamefont {Cover}\ and\ \citenamefont
  {Thomas}(2012)}]{cover2012elements}%
  \BibitemOpen
  \bibfield  {author} {\bibinfo {author} {\bibfnamefont {T.~M.}\ \bibnamefont
  {Cover}}\ and\ \bibinfo {author} {\bibfnamefont {J.~A.}\ \bibnamefont
  {Thomas}},\ }\href {https://doi.org/10.1002/047174882X} {\emph {\bibinfo
  {title} {Elements of Information Theory}}}\ (\bibinfo  {publisher} {John
  Wiley \& Sons},\ \bibinfo {year} {2012})\BibitemShut {NoStop}%
\bibitem [{\citenamefont {Bennett}\ \emph {et~al.}(1996)\citenamefont
  {Bennett}, \citenamefont {DiVincenzo}, \citenamefont {Smolin},\ and\
  \citenamefont {Wootters}}]{bennett1996mixed}%
  \BibitemOpen
  \bibfield  {author} {\bibinfo {author} {\bibfnamefont {C.~H.}\ \bibnamefont
  {Bennett}}, \bibinfo {author} {\bibfnamefont {D.~P.}\ \bibnamefont
  {DiVincenzo}}, \bibinfo {author} {\bibfnamefont {J.~A.}\ \bibnamefont
  {Smolin}},\ and\ \bibinfo {author} {\bibfnamefont {W.~K.}\ \bibnamefont
  {Wootters}},\ }\bibfield  {title} {\bibinfo {title} {Mixed-state entanglement
  and quantum error correction},\ }\href
  {https://doi.org/10.1103/PhysRevA.54.3824} {\bibfield  {journal} {\bibinfo
  {journal} {Physical Review A}\ }\textbf {\bibinfo {volume} {54}},\ \bibinfo
  {pages} {3824} (\bibinfo {year} {1996})}\BibitemShut {NoStop}%
\bibitem [{\citenamefont {Wilde}(2013)}]{wilde2013quantum}%
  \BibitemOpen
  \bibfield  {author} {\bibinfo {author} {\bibfnamefont {M.~M.}\ \bibnamefont
  {Wilde}},\ }\href {https://doi.org/10.1017/CBO9781139525343} {\emph {\bibinfo
  {title} {Quantum information theory}}}\ (\bibinfo  {publisher} {Cambridge
  University Press},\ \bibinfo {year} {2013})\BibitemShut {NoStop}%
\end{thebibliography}%

\onecolumngrid
\newpage

\begin{appendix}

	\section{Details about the proof of Theorem~\ref{th:upper}}

	In this section, we provide details in the derivation of Eq.~\eqref{eq:PauliDistribution}, \textit{i.e.}, the distribution of measurement outcomes at Line~5 in Algorithm~\ref{alg:main}. For clarity, denote the $k$-qubit Hilbert space of the ancilla as $A$, and divide the $n$-qubit Hilbert space of the main system into a $k$-qubit subspace $B$ and an ($n-k$)-qubit subspace $C$.
	The input state to the channel is
	$
		\ket{\Psi_0}_{AB}\otimes\ket{\phi_0^{\sf S}}_C,
	$
	and the measurement basis is
	$
		\{\ket{\Psi_v}_{AB}\otimes\ket{\phi_e^{\sf S}}_C\}_{v,e}.
	$
	Here, $\ket{\Psi_v}_{AB}$ are Bell states and can be expressed as
	\begin{equation}
		\begin{aligned}
			\ketbra{\Psi_v}{\Psi_v} &\coleq (P_v\otimes I) \ketbra{\Psi^+}{\Psi^+} (P_v\otimes I)\\
			&= \frac{1}{4^k}\sum_{u\in\mbb Z_2^{2k}} P_v P_u P_v \otimes P_u^{T}\\
			&= \frac{1}{4^k}\sum_{u\in\mbb Z_2^{2k}}(-1)^{\expval{u,v}}P_u \otimes P_u^{T},
		\end{aligned}
	\end{equation} 
	And the stabilizer state $\ket{\phi_e^{\sf S}}_C$ is defined as
	\begin{equation}
		\ketbra{\phi_e^{\sf S}}{\phi_e^{\sf S}}\coleq \frac{1}{2^{n-k}}\sum_{s\in{\sf S}}(-1)^{\expval{s,e}}P_s,
	\end{equation}
	for $e\in{\sf S}^\perp \coleq \mbb Z_2^{2(n-k)}/{\sf S}$. We remark that, the stabilizer state is well-defined for all $e\in\mbb Z_2^{2(n-k)}$, but $\ket{\phi_a^{\sf S}}$ and $\ket{\phi_b^{\sf S}}$ represent the same state if $a+b\in{\sf S}$ (bitwise modulo 2 sum), so we only need to consider the quotient space of $\mbb Z_2^{2(n-k)}$ over ${\sf S}$. When calculating the sympletic inner product $\expval{s,e}$, one should understand $e$ as an arbitrary representative of the coset it stands for.

	Therefore, the measurement outcome distribution can be calculated as
	\begin{equation}
	\begin{aligned}
		p(v,e) &= \Tr\left((\ketbra{\Psi_v}{\Psi_{v}}_{AB}\otimes\ketbra{\phi^{\sf S}_e}{\phi^{\sf S}_e}_{C})\mathds 1^A\otimes\Lambda^{BC}(\ketbra{\Psi_{0}}{\Psi_{0}}_{AB}\otimes\ketbra{\phi^{\sf S}_0}{\phi^{\sf S}_0}_{C})\right)\\
		&= \frac{1}{4^{n+k}}\Tr\left(\sum_{u,u'\in \mbb Z^{2k}_2}\sum_{s,s'\in \mbb {\sf S}} \lambda_{u'\oplus s'} \left((-1)^\expval{u,v}P_u\otimes P_u^T\otimes (-1)^{\expval{s,e}}P_s\right)_{ABC}\left( P_{u'}\otimes P_{u'}^T\otimes P_{s'} \right)_{ABC}\right)\\
		&= \frac{1}{2^{n+k}}\sum_{u\in \mbb Z^{2k}_2}\sum_{s\in \mbb {\sf S}}\lambda_{u\oplus s}(-1)^{\expval{u,v}}(-1)^{\expval{s,e}},
	\end{aligned}
	\end{equation}
	which is exactly Eq.~\eqref{eq:PauliDistribution} in the main text.
    It is then obvious that $p(v,e)$ and $\lambda_{u\oplus s}$ are related by the Walsh-Hadamard transform (see~\cite[Lemma~4]{flammia2020efficient}).

	\section{Proof of the lower bounds}\label{sec:lower}

	\subsection{A tight lower bound for non-adaptive and non-concatenating strategies}\label{app:lower1}
	
	In this section, we prove a matching lower bound of $\Omega(n2^{n-k})$ for all non-adaptive and non-concatenating k-qubit ancilla-assisted Pauli channel estimation protocols, which include the ancilla-free strategies ($k=0$) and $n$-qubit ancilla assisted strategies ($k=n$) as two special cases. 
	Recall that, by a non-adaptive and non-concatenating protocol we mean that, for each sample of the Pauli channel $\Lambda$, we prepare an $n+k$ qubits state, input it to $\Lambda\otimes \mathds1$, and apply a POVM measurement on the joint output state. The input state and measurement setting for the $i$th sample is not allowed to depend on previous measurement outcomes, nor do we allow concatenating multiple samples of $\Lambda$ in a single round of measurement (which is used in RB-type Pauli error estimation protocols~\cite{flammia2020efficient}). 
	
	\begin{theorem}\label{th:lower}
	For any non-adaptive, non-concatenating $k$-qubit ancilla-assisted protocols that give an estimate $\widehat{\bm{\lambda}}$ of the Pauli eigenvalues $\bm\lambda$ of an arbitrary unknown $n$-qubit Pauli channel such that
	\begin{equation}\label{eq:thm1}
		|\widehat\lambda_a - {\lambda_a}|< \frac12,\quad \forall a\in\mbb Z_2^{2n}
	\end{equation}
	holds with high probability, the number of samples of $\Lambda$ required is at least $\Omega(n2^{n-k})$.
	\end{theorem}

	Our proof techniques generalize the information-theoretical arguments by Huang \textit{et al.}~\cite{huang2021information}, which in turn stem from previous work on sample complexity lower bounds for quantum tomography~\cite{flammia2012quantum,haah2017sample,roth2018recovering}. Consider a communication protocol between Alice and Bob where their share the following ``codebook'':
	\begin{equation}
		(a,s)\in\{1,\cdots,4^n-1\}\times\{\pm 1\} ~\longrightarrow~ \Lambda_{(a,s)}(\cdot) = \frac{1}{2^n}\left(I\Tr(\cdot) + s P_a\Tr(P_a(\cdot))\right).
	\end{equation}
	Now, Alice picks at uniform random one out of the $2(4^n-1)$ possible pairs of $(a,s)$ and then send $N$ copies of the channel $\Lambda_{(a,s)}$ to Bob. If there exists a Pauli channel estimation protocol using $N$ samples and satisfying the assumption of Theorem~\ref{th:lower}, Bob can use that protocol to uniquely determine Alice's choice of $(a,s)$ with high probability, since the Pauli eigenvalues of any $\Lambda_{(a,s)}$ only take values from $\{-1,0,+1\}$. Suppose Bob's input state and POVM outcome for the $i$th sample is $\{\rho_i,E_i\}$. According to Fano's inequality, the mutual information between the random variable pair $(a,s)$ and Bob's measurement results has the following lower bound
	\begin{equation}
		I\left((a,s):\{\rho_1,E_1\},...,\{\rho_N,E_N\}\right) \ge \Omega(\log(2(4^n-1))) = \Omega(n).
	\end{equation}
	We also know by assumption that the measurement outcomes $\{\rho_i,E_i\}$ are independent from each other, conditioned on $(a,s)$. The chain rule of mutual information then gives that
	\begin{equation}
		\sum_{i=1}^N I\left((a,s),\{\rho_i,E_i\}\right) = I((a,s):\{\rho_1,E_1\},...,\{\rho_N,E_N\}) \ge \Omega(n).
	\end{equation}
	We will show that $I\left((a,s):\{\rho_i,E_i\}\right)\le \mc O(2^{k-n})$ in the following lemma. This would then give us the desired sample complexity lower bound $N\ge \Omega(n2^{n-k})$, which completes the proof of Theorem~\ref{th:lower}. 
	
	\begin{lemma}\label{le:mutual1}
	$
	I\left((a,s):\{\rho_i,E_i\}\right)\le \dfrac{2^k}{2^n-1}.
	$
	\end{lemma}
	\begin{proof}
	First notice that it suffices to consider pure state input and rank-$1$ POVM measurements. The latter comes from the fact that every POVM measurement can be viewed as a coarse-graining of some rank-1 POVM measurement. To see the former, consider the underlying distribution $ p((a,s),\{\rho_i,E_i\}) = p(a,s)p(\{E_i,\rho_i\}|a,s)$. The mutual information $I((a,s):\{\rho_i,E_i\})$ is convex about $p(\{E_i,\rho_i\}|a,s)$ when fixing $p(a,s)$ (see \textit{e.g.},~\cite[Theorem~2.7.4]{cover2012elements}), thus using mixed state input can never provide a larger mutual information.

	Thanks to this observation, we can without loss of generality let the input state be $\ket{A}$ and let the POVM measurement be $\{w_j2^{n+k}\ketbra{B_j}{B_j}\}_j$, where $\ket{A},\ket{B_j}\in\mbb C^{2^n\times 2^k}$ are unit vectors, and $\sum_j w_j = 1$ by normalization. We also abuse notations a little bit to let $A$ and $B_j$ denote the $2^n\times 2^k$ matrices that satisfy
	\begin{equation}\label{eq:matrixAB}
		\ket{A} = \sum_{p=0}^{2^n-1}\sum_{q=0}^{2^k-1}\expval{p|A|q}\ket{p}\ket{q},\quad     \ket{B_j} = \sum_{p=0}^{2^n-1}\sum_{q=0}^{2^k-1}\expval{p|B_j|q}\ket{p}\ket{q},
	\end{equation}
	where $\{\ket p\}$ and $\{\ket q\}$ are computational basis states. The normalization condition of $\ket{A}$ and $\ket{B_j}$ is equivalent to
	\begin{equation}
		\Tr(A^\dagger A) = \Tr(B_j^\dagger B_j) = 1.
	\end{equation}
	We also define $C_j\coleq B_jA^\dagger$ which is a $2^n\times 2^n$ matrix of rank less or equal to $2^k$.
	
	\medskip

	\noindent The mutual information between $(a,s)$ and a single round of measurement outcome $j$ can be upper bounded as
	\begin{equation}
	\begin{aligned}
		I((a,s):j) &= H(j) - H(j|a,s)\\
		&= -\sum_j\left(\mathop{\mbb{E}}_{(a,s)}p(j|a,s)\right)\log\left(\mathop{\mbb{E}}_{(a,s)}p(j|a,s)\right) + \mathop{\mbb{E}}_{(a,s)}\sum_j p(j|a,s)\log p(j|a,s)\\
		&\le \sum_j\frac{\mathop{\mbb{E}}_{(a,s)}[p(j|a,s)^2]-\mathop{\mbb{E}}_{(a,s)}[p(j|a,s)]^2}{\mathop{\mbb{E}}_{(a,s)}[p(j|a,s)]},
	\end{aligned}
	\end{equation}
	where the inequality follows from the fact that $\log(x)\le\log(y)+\frac{x-y}{y}$ in which we take $x\coleq p(j|a,s)$ and $y\coleq \mathop{\mbb E}_{(a,s)}[p(j|a,s)]$.
	
	\noindent The conditional probability $p(j|a,s)$ can be calculated as
	\begin{equation}
	\begin{aligned}
		p(j|a,s) &= w_j2^{n+k}\bra{B_j}\Lambda_{(a,s)}\otimes\mathds 1(\ketbra{A}{A})\ket{B_j}\\
		&= w_j\sum_{b=0}^{4^k}\left(\expval{B_j|I\otimes P_b|B_j}\expval{A|I\otimes P_b|A}+s\expval{B_j|P_a\otimes P_b|B_j}\expval{A|P_a\otimes P_b|A}\right)\\
		&= w_j\sum_{b=0}^{4^k}\left(\Tr(B_j^\dagger B_jP_b^T)\Tr(A^\dagger AP_b^T)+s\Tr(B_j^\dagger P_aB_jP_b^T)\Tr(A^\dagger P_aAP_b^T)\right) \\
		&= w_j2^k\left(\Tr(B_j^\dagger B_jA^\dagger A) + s\Tr(B_j^\dagger P_a B_jA^\dagger P_a A)\right)\\
		&= w_j2^k\left( \Tr(C_j^\dagger C_j)+s\Tr(P_aC_jP_aC_j^\dagger) \right),
	\end{aligned}
	\end{equation}
	where we expand the identity channel as $\mathds 1(\cdot) = 2^{-k}\sum_{b=0}^{4^k}P_b\Tr(P_b(\cdot))$ in the second line, and use the fact $2^{-k}\sum_{b=0}^{4^k}P_b\otimes P_b$ equals to the swap operator in the fourth line. 
	
	\medskip
	
	\noindent The average value and second moment of $p(j|a,s)$ according to the distribution of $(a,s)$ are
	\begin{equation}
	\begin{aligned}
		\mathop{\mbb{E}}_{(a,s)}[p(j|a,s)] &= w_j2^k \Tr(C_j^\dagger C_j), \\
		\mathop{\mbb{E}}_{(a,s)}[p(j|a,s)^2] &= w_j^2 4^k \left(\Tr^2(C_j^\dagger C_j) + \frac{1}{4^n-1}\sum_{a=1}^{4^n-1}\Tr^2(P_aC_jP_aC_j^\dagger) \right).
	\end{aligned}
	\end{equation}
	Hence we have the following bound for the mutual information
	\begin{equation}\label{eq:mutual1}
		I\left((a,s):j\right) \le \sum_j w_j2^k\Tr(C^\dagger_jC_j)\left(\frac{1}{4^n-1}\sum_{a=1}^{4^n-1}\frac{\Tr(P_aC^\dagger_jP_aC_j)^2}{\Tr(C^\dagger_jC_j)^2}\right).
	\end{equation}
	Now we further calculate the R.H.S. of the above inequality. 
	Let $M=P_aC_j^\dagger P_aC_j$. Notice that
	\begin{equation}
		\text{rank}(M) \le \text{rank}(C_j) \le 2^k,
	\end{equation}
	which means there exists a rank-$2^k$ projector $\Pi$ such that $\Tr(M) = \Tr(M\Pi)$. According to Cauchy-Schwarz inequality,
	\begin{equation}\label{eq:matrixCS}
		\begin{aligned}
			\Tr(M)^2 & = \Tr(M\Pi)^2 \\
			& \le \Tr(MM^\dagger)\Tr(\Pi^\dagger\Pi)\\
			& = \Tr(C_jC^\dagger_jP_aC_jC_j^\dagger P_a )\times 2^k		
		\end{aligned}
	\end{equation}
	Substitute this into Eq.~\eqref{eq:mutual1},
	\begin{equation}\label{eq:mutualbound}
		\begin{aligned}
			I((a,s):j) &\le \sum_j w_j2^k\Tr(C^\dagger_jC_j)\left(\frac{2^k}{4^n-1}\sum_{a=1}^{4^n-1}\frac{\Tr(C_jC^\dagger_jP_aC_jC_j^\dagger P_a )	}{\Tr(C^\dagger_jC_j)^2}\right) \\
			&= \sum_j w_j2^k\Tr(C^\dagger_jC_j)\left(\frac{2^k}{4^n-1}\times     
			\frac{\sum_{a=0}^{4^n-1}\Tr(C_jC^\dagger_jP_aC_jC_j^\dagger P_a )-\Tr(C_jC_j^\dagger C_jC_j^\dagger)}{\Tr(C^\dagger_jC_j)^2}\right)\\
			&= \sum_j w_j2^k\Tr(C^\dagger_jC_j)\left(\frac{2^k}{4^n-1}\times     
			\frac{2^n\Tr(C_j^\dagger C_j)^2-\Tr(C_jC_j^\dagger C_jC_j^\dagger)}{\Tr(C^\dagger_jC_j)^2}\right)\\
			&\le \sum_j w_j2^k\Tr(C^\dagger_jC_j)\frac{2^k}{2^n-1}\\
			&= \frac{2^k}{2^n-1},
		\end{aligned}
	\end{equation}
	where the third line uses the following formula of Pauli twirling,
	\begin{equation}
	    \frac{1}{4^n}\sum_{a=0}^{4^n} P_a X P_a = \frac{1}{2^n}\Tr(X)I,
	\end{equation}
	and the last line follows from the fact that 
	\begin{equation}
		\sum_j w_j2^k\Tr(C^\dagger_jC_j)=\sum_j\mathop{\mbb{E}}_{(a,s)}[p(j|a,s)] = 1.
	\end{equation}
	This completes the proof of Lemma~\ref{le:mutual1}.
	\end{proof}
	
	\subsection{A lower bound for adaptive but non-concatenating strategies}
	
	In this section, we prove a (perhaps loose) lower bound of $\Omega(2^{(n-k)/3})$ for all adaptive and non-concatenating $k$-qubit ancilla-assisted Pauli channel estimation protocols, as stated in the following theorem.
	
	\begin{theorem}\label{th:adalower}
		For any adaptive, non-concatenating $k$-qubit ancilla-assisted protocols that give an estimate $\widehat{\bm{\lambda}}$ of the Pauli eigenvalues $\bm\lambda$ of an arbitrary unknown $n$-qubit Pauli channel such that
		\begin{equation}
			|\widehat\lambda_a - {\lambda_a}|< \frac12,\quad \forall a\in\mbb Z_2^{2n}
		\end{equation}
		holds with high probability, the number of samples of $\Lambda$ required is at least $\Omega(2^{(n-k)/3})$.
		\end{theorem}
		
	\noindent Our proof techniques generalize the methods of Huang \textit{et al.}~\cite{huang2021information} for proving adaptive sample complexity lower bounds for Pauli expectation values estimation of unknown quantum states. 
	Consider the following $4^n$ possible Pauli channels 
	\begin{equation}
		\left\{
		\begin{aligned}
			\Lambda_\text{dep}(\cdot) &= \frac{1}{2^n}I\Tr(\cdot),\\ 
			\Lambda_a(\cdot) &= \frac{1}{2^n}\left(I\Tr(\cdot)+P_a\Tr(P_a(\cdot))\right),~\forall a\in\{1,\cdots,4^n-1\}.
		\end{aligned}
		\right.
	\end{equation}
	Here $\Lambda_\text{dep}$ is known as the completely deplorizing channel.
	If there exists an Pauli channel estimation protocol satisfying the requirement of Theorem~\ref{th:adalower}, one can unambiguously identify each one of the $4^n$ possible Pauli channels appearing above with high probability, given sufficient number of samples of $\Lambda$. This in turn implies that one should be able to distinguish the following two equal-probable hypotheses with high success probability.
	\begin{enumerate}
		\item Given $N$ copies of $\Lambda = \Lambda_\text{dep}$.
		\item Given $N$ copies of $\Lambda = \Lambda_{a}$ for a uniformly-randomly picked $a\in\{1,\cdots,4^n-1\}$.
	\end{enumerate}
	For an adaptive but non-concatenating protocol, one has to choose an $2^n\times 2^k$ dimensional input state and a POVM measurement for the $i$th sample of $\Lambda$, where the choice may depend on previous measurement outcomes. Denote the measurement outcome of the $i$th round as $o_i$. We explicitly write the state and measurement at the $i$th round as $\rho^{o_{< i}}$ and $\{E_{j}^{o_{< i}}\}_j$ to emphasize their dependence on $o_{<i}\coleq[o_1,...,o_{i-1}]$.Denote the measurement outcomes among all the $N$ samples as $o_{1:N}\coleq[o_1,\cdots,o_N]$. 
	The probability distribution of $o_{1:N}$ under the above two hypothesis can be expressed as
	\begin{equation}
		\left\{
		\begin{aligned}
			\text{Hypothesis 1:}\quad p_1(o_{1:N}) &= \prod_{i=1}^N\Tr\left( E_{o_i}^{o_{<i}}\Lambda_\text{dep}\otimes\mathds 1(\rho^{o_{<i}}) \right),\\
			\text{Hypothesis 2:}\quad p_2(o_{1:N}) &= \mathop{\mbb E}_{a\ne 0}\prod_{i=1}^N\Tr\left( E_{o_i}^{o_{<i}}\Lambda_a\otimes\mathds 1(\rho^{o_{<i}}) \right).
		\end{aligned}
		\right.
	\end{equation}
	The ability to distinguish these two hypotheses is equivalent to the ability to distinguish $p_1$ from $p_2$. The maximal success probability of distinguishing two probability distributions is given by $\frac12(1+\mr{TV}(p_1,p_2))$ where $\mr{TV}$ stands for the total variance distance defined as follows
	\begin{equation}
		\mr{TV}(p_1,p_2) \coleq  \sum_{\substack{
					o_{1:N}~\text{s.t.}\\
					p_1(o_{1:N})\ge p_2(o_{1:N})
				}}(p_1(o_{1:N}) - p_2(o_{1:N})).
	\end{equation}
	
	We will show in the following Lemma that $\mr{TV}(p_1,p_2)=\mc O(N2^{(k-n)/3})$, which immediately implies that one must have $N=\Omega(2^{(n-k)/3})$ in order to obtain high success probability in distinguishing $p_1$ from $p_2$. This then completes the proof of Theorem~\ref{th:adalower}.
	
	\begin{lemma}\label{le:TV}
		$\mathrm{TV}(p_1,p_2) \le 2N\left(\cfrac{2^k}{2^n-1}\right)^{1/3}$.
	\end{lemma}
	
	\begin{proof}
		First notice that it suffices to consider pure state input and rank-$1$ POVM measurement. 
		For the former, the probability distribution obtained from mixed state input can be viewed as a convex combination of distributions obtained from pure state input. Thanks to the joint convexity, mixed state input can not yield a larger total variance distance;
		For the latter, every POVM measurement can be viewed as a coarse-graining of some rank-$1$ POVM measurement. Because of the data-processing property, this coarse-graining would not yield a larger total variance distance.

		In light of this observation, we can without loss of generality let the input state at the $i$th round be $\ket{A^{o_{<i}}}$ and let the POVM measurement be $\{w_{{o_i}}^{o_{<i}}2^{n+k}\ketbra{B_{{o_i}}^{o_{<i}}}{B_{{o_i}}^{o_{<i}}}\}_{o_i}$, conditioned on previous measurement outcomes $o_{<i}$, where $\ket{A^{o_{<i}}},\ket{B_{{o_i}}^{o_{<i}}}\in\mbb C^{2^n\times 2^k}$ are unit vectors, and $\sum_{o_i} w_{{o_i}}^{o_{<i}} = 1$ by normalization. 
		We also introduce the two $2^n\times 2^k$ matrices $A^{o_{<i}},~B_{{o_i}}^{o_{<i}}$ defined similarly as in Eq.~\eqref{eq:matrixAB}, and define $C^{o_{<i}}_{{o_i}}\coleq B^{o_{<i}}_{{o_i}}A^{o_{<i}\dagger}$ which is a $2^n\times 2^n$ matrix of rank less or equal to $2^k$.
		With the above definitions, one can verify that $p_1$ and $p_2$ can be expressed as follows
		\begin{equation}
			\begin{aligned}
				p_1(o_{1:N}) &= \prod_{i=1}^N w_{o_i}^{o_{<i}}2^{n+k} \bra{B_{o_i}^{o_{<i}}}\Lambda_\text{dep}\otimes\mathds 1\left((\ketbra{A^{o_{<i}}}{A^{o_{<i}}}\right)\ket{B_{o_i}^{o_{<i}}} ,\\
				&=  \prod_{i=1}^N w_{o_i}^{o_{<i}}2^{k}\Tr\left( C_{o_i}^{o_{<i}\dagger}C_{o_i}^{o_{<i}} \right) \\
				p_2(o_{1:N}) &= \mathop{\mbb{E}}_{a\ne 0}\prod_{i=1}^N w_{o_i}^{o_{<i}}2^{n+k} \bra{B_{o_i}^{o_{<i}}}\Lambda_a\otimes\mathds 1\left((\ketbra{A^{o_{<i}}}{A^{o_{<i}}}\right)\ket{B_{o_i}^{o_{<i}}}\\
				&=\mathop{\mbb{E}}_{a\ne 0}\prod_{i=1}^N w_{o_i}^{o_{<i}}2^{k}\left(\Tr\left( C_{o_i}^{o_{<i}\dagger}C_{o_i}^{o_{<i}} \right) + \Tr\left( C_{o_i}^{o_{<i}\dagger}P_a C_{o_i}^{o_{<i}} P_a\right) \right).
			\end{aligned}
		\end{equation}
		The total variance between $p_1$ and $p_2$ can then be bounded as
		\begin{equation}\label{eq:TV1}
			\begin{aligned}
				\mr{TV}(p_1,p_2) &= \mathop{{}	\mbb{E}}_{a\ne 0} \sum_{\substack{
					o_{1:N}~\text{s.t.}\\
					p_1(o_{1:N})\ge p_2(o_{1:N})
				}}
				\left( \prod_{i=1}^N w_{o_i}^{o_{<i}}2^k\Tr\left( C_{o_i}^{o_{<i}\dagger}C_{o_i}^{o_{<i}} \right) \right)
				\left( 1-\prod_{i=1}^N\left(
					1+\frac{\Tr\left( C_{o_i}^{o_{<i}\dagger}P_a C_{o_i}^{o_{<i}} P_a\right) }{\Tr\left( C_{o_i}^{o_{<i}\dagger} C_{o_i}^{o_{<i}} \right) }
				\right) \right)  \\
				& = \mathop{\mbb{E}}_{a\ne 0}\sum_{\substack{
					o_{1:N}~\text{s.t.}\\
					p_1(o_{1:N})\ge p_2(o_{1:N})
				}}
				p_1(o_{1:N})
				\left( 1-\prod_{i=1}^N\left(
					1+\frac{\Tr\left( C_{o_i}^{o_{<i}\dagger}P_a C_{o_i}^{o_{<i}} P_a\right) }{\Tr\left( C_{o_i}^{o_{<i}\dagger} C_{o_i}^{o_{<i}} \right) }
				\right) \right),
			\end{aligned}
		\end{equation}
		In order to bound the R.H.S., we make use of a technique from Huang~\textit{et al.}~\cite{huang2021information}.
		Let $C$ denote an arbitrary $2^n\times 2^n$ complex matrix of rank no more than $2^k$, consider the following subset of $n$-qubit Pauli operators
		\begin{equation}\label{eq:Gsize}
			G \coleq \left\{ a\in\{1,\cdots,4^n-1\}~:~ \left|
			\frac{\Tr\left( C^\dagger P_a C P_a\right) }{\Tr\left( C^\dagger C \right) }
			\right| \le \left(\frac{2^k}{2^n-1}\right)^{1/3}\right\}.
		\end{equation} 
		We claim that the size of $G$ satisfies
		\begin{equation}
			\left|G\right| \ge \left( 1- \left(\frac{2^k}{2^n-1}\right)^{1/3}\right)(4^n-1),
		\end{equation}
		which can be shown by contradiction: Suppose this does not hold, we would have
		\begin{equation}
			\begin{aligned}
				\sum_{a=1}^{4^n-1} \left(\frac{\Tr\left( C^\dagger P_a C P_a\right) }{\Tr\left( C^\dagger C \right) }\right)^2 
				&\ge  \left(\frac{2^k}{2^n-1}\right)^{2/3}\times \left( 4^n-1- \left|G\right| \right)\\
				&>  \left(\frac{2^k}{2^n-1}\right)^{2/3}\times \left(\frac{2^k}{2^n-1}\right)^{1/3}(4^n-1)\\
				&= 2^k(2^n+1).
			\end{aligned}
		\end{equation}
		However, the L.H.S. of the above can be upper bounded as
		\begin{equation}
			\begin{aligned}
				\sum_{a=1}^{4^n-1} \left(\frac{\Tr\left( C^\dagger P_a C P_a\right) }{\Tr\left( C^\dagger C \right) }\right)^2 
				&\le \sum_{a=1}^{4^n-1} \frac{2^k\Tr(CC^\dagger P_a CC^\dagger P_a)}{\Tr^2\left( C^\dagger C \right)}\\
				&= 2^k\frac{2^n\Tr^2\left( C^\dagger C \right) - \Tr(C^\dagger CC^\dagger C)}{\Tr^2\left( C^\dagger C \right)}\\
				&\le 2^{k}2^n,
			\end{aligned}
		\end{equation}
		where the first inequality follows from Cauchy-Schwarz (see Eq.~\eqref{eq:matrixCS}), and the first equality evaluates the Pauli twirling (see Eq.~\eqref{eq:mutualbound}). This leads to contradition, and hence proves the desired lower bound on $|G|$. 
		
		Now, we define another subset of $n$-qubit Pauli operators, conditioned on the measurement outcomes $o_{1:N}$, as follows
		\begin{equation}
			G^{(o_{1:N})} \coleq \left\{ a\in\{1,\cdots,4^n-1\}~:~ \left|
			\frac{\Tr\left( C_{o_i}^{o_{<i}\dagger}P_a C_{o_i}^{o_{<i}} P_a\right) }{\Tr\left( C_{o_i}^{o_{<i}\dagger} C_{o_i}^{o_{<i}} \right)}
			\right| \le \left(\frac{2^k}{2^n-1}\right)^{1/3},~\forall i = 1,\cdots,N\right\}.
		\end{equation}
		By applying a ``union bound'' on Eq.~\eqref{eq:Gsize}, we immediately have the following lower bound on the size of $G^{(o_{1:N})}$:
		\begin{equation}\label{eq:Gsize2}
			\left|G^{(o_{1:N})}\right| \ge  \left( 1-N \left(\frac{2^k}{2^n-1}\right)^{1/3}\right)(4^n-1).
		\end{equation}
	
		We are now ready to upper bound $\mr{TV}(p_1,p_2)$ from Eq.~\eqref{eq:TV1}. The strategy is to divide the sum over all Pauli operators into $G^{(o_{1:N})}$ and ${\sf P}^n\backslash G^{(o_{1:N})}$. All terms within the former group are small thanks to the definition of $G^{(o_{1:N})}$; Terms within the latter group could be large, but the total number of them are small. Combining these two gives a pretty good upper bound on $\mr{TV}(p_1,p_2)$. 
		In math,
		\begin{equation}
			\begin{aligned}
				\mr{TV}(p_1,p_2) &= \frac{1}{4^n-1}\sum_{\substack{
					o_{1:N}~\text{s.t.}\\
					p_1(o_{1:N})\ge p_2(o_{1:N})
				}}
				p_1(o_{1:N})
				\left(
					\sum_{a\in  G^{(o_{1:N})}} + 
					\sum_{\substack{
						a\in  \mbb Z^{2n}_2\backslash G^{(o_{1:N})},\\
					a\ne 0
					}}
				\right)
				\left( 1-\prod_{i=1}^N\left(
					1+\frac{\Tr\left( C_{o_i}^{o_{<i}\dagger}P_a C_{o_i}^{o_{<i}} P_a\right) }{\Tr\left( C_{o_i}^{o_{<i}\dagger} C_{o_i}^{o_{<i}} \right) }
				\right) \right)\\
				&\le \frac{1}{4^n-1}\sum_{\substack{
					o_{1:N}~\text{s.t.}\\
					p_1(o_{1:N})\ge p_2(o_{1:N})
				}}p_1(o_{1:N})
				\left(
					\left|G^{(o_{1:N})}\right|\times\left( 1-\left(1- \left(\frac{2^k}{2^n-1}\right)^{1/3}  \right)^N \right)
					+
					\left(4^n-1-\left|G^{(o_{1:N})}\right|\right)
				\right)\\
				&\le \sum_{\substack{
					o_{1:N}~\text{s.t.}\\
					p_1(o_{1:N})\ge p_2(o_{1:N})
				}}p_1(o_{1:N})
				\left(
					N\left( \frac{2^k}{2^n-1} \right)^{1/3} + 
					N\left( \frac{2^k}{2^n-1} \right)^{1/3}
				\right)\\
				& \le 2N\left( \frac{2^k}{2^n-1} \right)^{1/3}.
			\end{aligned}
		\end{equation}
		The first inequality uses an additional fact that $\left| \Tr(C^\dagger P_aCP_a) \right| / \left| \Tr(C^\dagger C) \right| \le 1$ to bound the second sum, which follows from the Cauchy-Schwarz inequality. 
		The second inequality uses the bounds $\left| G^{(o_{1:N})} \right| \le 4^n - 1$ for the first sum and Eq.~\eqref{eq:Gsize2} for the second sum, as well as the fact that $1-(1-x)^N\le Nx$ for all $0\le x\le 1$. (Note that $2^k/(2^n-1)\le 1$ only if $k\le n-1$, but our targeted upper bound trivially holds for $k=n$.) This completes the proof of Lemma~\ref{le:TV}.
	
	\end{proof}

	\subsection{A lower bound for the most general ancilla-free strategies}

	The lower bounds in the previous sections work for non-concatenating strategies, where one is not allowed to concatenate (or, coherently access) multiple copies of the unknwon channel before doing a single measurement. 
	The concatenating strategies, as depicted in Fig.~\ref{fig:model}, are however a natural apparatus for many randomized benchmarking protocols, where one effectively concatenate a varying number of noise channels in order to measure a series of exponentially-decaying values, and then extract the parameters of interest via fitting the decay rate.
	The original purpose of such concatenation in these protocols is to eliminate the effect of state-preparation-and-measurement (SPAM) error. Here, we want to understand whether concatenating strategies also provide a sample complexity advantage. The short answer is no. We will present a sample complexity lower bound of $\Omega(2^{n/3})$ for the most powerful (adaptive, concatenating) ancilla-free ($k=0$) Pauli channel estimation schemes. This result justifies our claim that ancillary systems are indeed indispensable to overcome the exponential barrier in sample complexity. 
	
	To start with, we give a rigorous definition of the most general ancilla-free strategies that we are going to study
	\begin{definition}\label{de:no_ancilla}
		Let $\Lambda$ be an unknwon $n$-qubit Pauli channel.
		An adaptive, concatenating, ancilla-free $(k=0)$ estimation protocol is specified by the following parameters. Let $N$ denote the total rounds of measurements. For the $i$th round, let $\rho^{o_{<i}}$ denote the input state, $\{E_{o_i}^{o_{<i}}\}_{o_i}$ denote the POVM measurement, $M^{o_{<i}}$ denote the length of concatenation, and $\{\mc C_k^{o_{<i}}\}_{k=1}^{M^{o_{<i}}}$ denote a set of processing channels. The $i$th measurement outcome is given by $o_i$ with probability
		\begin{equation}
			\Pr(o_i|o_{<i}) = \tr\left[E_{o_i}^{o_{<i}} \Lambda(\mc C_{M^{o_{<i}}-1}^{o_{<i}} ( \cdots  
			\mc C^{o_{<i}}_2(\Lambda(\mc C^{o_{<i}}_1(\Lambda(\rho^{o_{<i}}))))
			\cdots )  ) \right].
		\end{equation}
		All the superscript $o_{<i}\coleq [o_1,\cdots,o_{i-1}]$ are used to emphasize the dependence on previous measurement outcomes. The protocol should produce an estimate of $\Lambda$ via classical processing on the measurement outcomes $o_{1:N}$.
	\end{definition}
	
	The problem we are interested in is still approximating the Pauli eigenvalues $\bm \lambda$ to small error in $l_\infty$ distance. Note that, the parameter $N$ in the above definition is not exactly the sample complexity but is the total number of measurements conducted. Since one is allowed to concatenate multiple copies of $\Lambda$ in a single measurement, $N$ is a lower bound for the sample complexity of $\Lambda$. Our result is summarized in the following theorem.
	
	\begin{theorem}\label{th:no_ancilla}
		For any adaptive, concatenating, ancilla-free $(k=0)$ protocols that gives an estimate $\widehat{\bm{\lambda}}$ of the Pauli eigenvalues $\bm\lambda$ of an arbitrary unknown $n$-qubit Pauli channel $\Lambda$ such that
		\begin{equation}
			|\widehat\lambda_a - {\lambda_a}|< \frac12,\quad \forall a\in\mbb Z_2^{2n}
		\end{equation}
		holds with high probability, the rounds of measurements $N$ (and hence the number of samples of $\Lambda$) required is at least $\Omega(2^{n/3})$.
	\end{theorem}

	\begin{proof}
		The proof methods are similar to the proof of Theorem~\ref{th:adalower}. Define the following Pauli channels.
		\begin{equation}
			\left\{
			\begin{aligned}
				\Lambda_\text{dep}(\cdot) &= \frac{1}{2^n}I\Tr(\cdot),\\ 
				\Lambda_a(\cdot) &= \frac{1}{2^n}\left(I\Tr(\cdot)+P_a\Tr(P_a(\cdot))\right),~\forall a\in\{1,\cdots,4^n-1\}.
			\end{aligned}
			\right.
		\end{equation}
		We consider the problem of distinguishing the following two equal-probable hypotheses
		\begin{enumerate}
			\item Given $N$ copies of $\Lambda = \Lambda_\text{dep}$.
			\item Given $N$ copies of $\Lambda = \Lambda_{a}$ for a uniformly-randomly picked $a\in\{1,\cdots,4^n-1\}$.
		\end{enumerate}
		An estimation protocol satisfying our assumptions should be able to distinguish these two hypotheses with high probability. Let the probability distribution of the measurement outcomes $o_{1:N}$ under these two hypotheses be $p_1$ and $p_2$, respectively. The total variance distance $\mr{TV}(p_1,p_2)$ must be at least $\Omega(1)$ for the distinguishing task to succeed with high probability. We will show in the following that $\mr{TV}(p_1,p_2) = O(N2^{-n/3})$, which then gives the claimed lower bound of $N = \Omega(2^{n/3})$.
	
		To start with, based on the same argument as in the proof of Lemma~\ref{le:TV}, it suffices to consider pure state input and rank-$1$ POVM measurements, so we replace $\rho^{o_{<i}}$ and $\{E_{o_i}^{o_{<i}}\}_{o_i}$ in Definition~\ref{de:no_ancilla} with $\ketbra{A^{o_{<i}}}{A^{o_{<i}}}$ and $\{w^{o_{<i}}_{o_i}2^n \ketbra{B^{o_{<i}}}{B^{o_{<i}}} \}_{o_i}$ respectively, where $\ket{A^{o_{<i}}},\ket{B_{{o_i}}^{o_{<i}}}\in\mbb C^{2^n}$ are unit vectors, and $\sum_{o_i} w_{{o_i}}^{o_{<i}} = 1$ by normalization.
		
		Next, we calculate the distribution of $o_{1:N}$ under the two different hypotheses. 
		The expression for $p_1$ can be easily obtained, as $\Lambda_\text{dep}$ is simply the completely deplorizing channel. 
		We have
		\begin{equation}
		\begin{aligned}
			p_1(o_{1:N}) &= \prod_{i=1}^N 
			w_{o_i}^{o_{<i}}2^n\bra{B_{o_i}^{o_{<i}}} \Lambda_\mr{dep}(\mc C_{M^{o_{<i}}-1}^{o_{<i}} ( \cdots  
			\mc C^{o_{<i}}_2(\Lambda_\mr{dep}(\mc C^{o_{<i}}_1(\Lambda_\mr{dep}(
			\ketbra{A^{o_{<i}}}{A^{o_{<i}}}
			))))
			\cdots )  ) \ket{B_{o_i}^{o_{<i}}} \\
			&= \prod_{i=1}^{N}w_{o_i}^{o_{<i}}.
		\end{aligned}
		\end{equation}
		The expression for $p_2$ is more complicated. We first define the following recursive expression
		\begin{equation}
			\xi_a^{o_{<i}}[m]\coleq\left\{
			\begin{aligned}
				&2^{-n}\Tr\left(P_a\mc C_{m-1}(I+P_a\xi_a^{o_{<i}}[m-1])\right),\quad 2\le m\le M^{o_{<i}},\\
				&\bra{A^{o_{<i}}} P_a \ket{A^{o_{<i}}},\quad m=1.
			\end{aligned}
			\right.
		\end{equation}
		The expression for $p_2$ can then be calculated as follows
		\begin{equation}
			\begin{aligned}
				p_2(o_{1:N}) &= \mathop{\mbb E}_{a\ne 0}\prod_{i=1}^N 
				w_{o_i}^{o_{<i}}2^n\bra{B_{o_i}^{o_{<i}}} \Lambda_a(\mc C_{M^{o_{<i}}-1}^{o_{<i}} ( \cdots  
				\mc C^{o_{<i}}_2(\Lambda_a(\mc C^{o_{<i}}_1(\Lambda_a(
				\ketbra{A^{o_{<i}}}{A^{o_{<i}}}
				))))
				\cdots )  ) \ket{B_{o_i}^{o_{<i}}}\\
				&= \mathop{\mbb E}_{a\ne 0}\prod_{i=1}^N w_{o_i}^{o_{<i}}
				\bra{B_{o_i}^{o_{<i}}} \Lambda_a(\mc C_{M^{o_{<i}}-1}^{o_{<i}} ( \cdots\mc C^{o_{<i}}_2(\Lambda_a(\mc C^{o_{<i}}_1
				(
					I+P_a\xi_a^{o_{<i}}[1]
				)
				))\cdots))\ket{B_{o_i}^{o_{<i}}}\\
				&= \mathop{\mbb E}_{a\ne 0}\prod_{i=1}^N w_{o_i}^{o_{<i}}
				\bra{B_{o_i}^{o_{<i}}} \Lambda_a(\mc C_{M^{o_{<i}}-1}^{o_{<i}} ( \cdots\mc C^{o_{<i}}_2
				(
					I + 2^{-n}P_a\Tr(P_a\mc C^{o_{<i}}_1(I + P_a\xi_a^{o_{<i}}[1])
				)
				\cdots))\ket{B_{o_i}^{o_{<i}}}\\
				&= \mathop{\mbb E}_{a\ne 0}\prod_{i=1}^N w_{o_i}^{o_{<i}}
				\bra{B_{o_i}^{o_{<i}}} \Lambda_a(\mc C_{M^{o_{<i}}-1}^{o_{<i}} ( \cdots\mc C^{o_{<i}}_2
				(
					I + P_a\xi_a^{o_{<i}}[2]
				)
				\cdots))\ket{B_{o_i}^{o_{<i}}}\\
				&= \cdots\\
				&= \mathop{\mbb E}_{a\ne 0}\prod_{i=1}^N w_{o_i}^{o_{<i}}
				\bra{B_{o_i}^{o_{<i}}} 
				I + P_a\xi_a^{o_{<i}}[M^{o_{<i}}]
				\ket{B_{o_i}^{o_{<i}}}\\
				&= \mathop{\mbb E}_{a\ne 0}\prod_{i=1}^N w_{o_i}^{o_{<i}}(1 + \xi_a^{o_{<i}}[M^{o_{<i}}] \bra{B_{o_i}^{o_{<i}}}P_a\ket{B_{o_i}^{o_{<i}}}).
			\end{aligned}
		\end{equation}
		The third line uses the fact that $\mc C_k^{o_{<i}}$ is trace-preserving. The total variance distance between $p_1$ and $p_2$ is then
		\begin{equation}\label{eq:TV1'}
		\begin{aligned}
			\mr{TV}(p_1,p_2) &= \mathop{\mbb E}_{a\ne 0}\sum_{\substack{
				o_{1:N}~\text{s.t.}\\
				p_1(o_{1:N})\ge p_2(o_{1:N})
			}}
			\left(\prod_{i=1}^{N}w_{o_i}^{o_{<i}}\right)\left(1-\prod_{i=1}^{N}\left(
				1 + \xi_a^{o_{<i}}[M^{o_{<i}}] \bra{B_{o_i}^{o_{<i}}}P_a\ket{B_{o_i}^{o_{<i}}}
			\right)\right)\\
			&= \mathop{\mbb E}_{a\ne 0}\sum_{\substack{
				o_{1:N}~\text{s.t.}\\
				p_1(o_{1:N})\ge p_2(o_{1:N})
			}}
			p_1(o_{1:N})\left(1-\prod_{i=1}^{N}\left(
				1 + \xi_a^{o_{<i}}[M^{o_{<i}}] \bra{B_{o_i}^{o_{<i}}}P_a\ket{B_{o_i}^{o_{<i}}}
			\right)\right).
		\end{aligned}
		\end{equation}
		We now need a bound of $|\xi_a^{o_{<i}}[M^{o_{<i}}]|\le 1$, which can be shown by induction. We see $|\xi_a^{o_{<i}}[1]| \le 1$ by definition. Suppose $|\xi_a^{o_{<i}}[m-1]| \le 1$, we have
		\begin{equation}
			\begin{aligned}
				|\xi_a^{o_{<i}}[m]| &= \left|\Tr\left(P_a \mc C^{o_{<i}}_{m-1}\left(\frac{I +P_a\xi_a^{o_{<i}}[m-1]}{2^n}\right)\right)\right|\\
				&\le \|P_a\|_\infty \Tr \left|  \mc C^{o_{<i}}_{m-1}\left(\frac{I +P_a\xi_a^{o_{<i}}[m-1]}{2^n}\right)  \right|\\
				&=\Tr\left(\frac{I +P_a\xi_a^{o_{<i}}[m-1]}{2^n}\right)\\
				&=1.
			\end{aligned}
		\end{equation}
		The first line is by the defining recursive expression; The second line is by the tracial matrix H{\"o}lder inequality; The third line uses the fact that $\mc C_{m-1}^{o_{<i}}$ is a positive map, and that ${2^{-n}}\left({I +P_a\xi_a^{o_{<i}}[m-1]}\right)$ is positive semidefinite thanks to the induction hypothesis $|\xi_a^{o_{<i}}[m-1]|\le1$. Thus we can remove the modulus within the trace, and also remove $\mc C_{m-1}^{o_{<i}}$ as it is trace-preserving. By induction, we've shown $|\xi_a^{o_{<i}}[M^{o_{<i}}]| \le 1$.

		The remaining part of bounding $\mr{TV}(p_1,p_2)$ is basically the same as in the proof of Lemma~\ref{le:TV}. We repeat it here for completeness.
		Let $\ket B$ be any $n$-qubit pure state. Consider the following subset of $n$-qubit Pauli operators
		\begin{equation}\label{eq:Gsize'}
			G \coleq \left\{ a\in\{1,\cdots,4^n-1\}~:~ \left|
			\bra{B}P_a\ket{B}
			\right| \le \left(\frac{1}{2^n+1}\right)^{1/3}\right\}.
		\end{equation} 
		We claim that the size of $G$ satisfies
		\begin{equation}
			\left|G\right| \ge \left( 1- \left(\frac{1}{2^n+1}\right)^{1/3}\right)(4^n-1),
		\end{equation}
		which can be shown by contradiction: Suppose this does not hold, we would have
		\begin{equation}
			\begin{aligned}
				\sum_{a=1}^{4^n-1} \bra B P_a \ket B^2
				&\ge  \left(\frac{1}{2^n+1}\right)^{2/3}\times \left( 4^n-1- \left|G\right| \right)\\
				&>  \left(\frac{1}{2^n+1}\right)^{2/3}\times \left(\frac{1}{2^n+1}\right)^{1/3}(4^n-1)\\
				&= 2^n-1.
			\end{aligned}
		\end{equation}
		However, the L.H.S. of the above can be calculated as
		\begin{equation}
			\begin{aligned}
				\sum_{a=1}^{4^n-1} \bra B P_a \ket B^2
				= \sum_{a=0}^{4^n-1} \bra B P_a \ket B^2 - 1
				= 2^n-1.
			\end{aligned}
		\end{equation}
		This leads to contradition, and hence proves the desired lower bound on $|G|$. 
		
		Now, we define another subset of $n$-qubit Pauli operators, conditioned on the measurement outcomes $o_{1:N}$, as follows
		\begin{equation}
			G^{(o_{1:N})} \coleq \left\{ a\in\{1,\cdots,4^n-1\}~:~ \left|
			\bra{B_{o_i}^{o_{<i}}}P_a\ket{B_{o_i}^{o_{<i}}}
			\right| \le \left(\frac{1}{2^n+1}\right)^{1/3},~\forall i = 1,\cdots,N\right\}.
		\end{equation}
		By applying a ``union bound'' on Eq.~\eqref{eq:Gsize'}, we immediately have the following lower bound on the size of $G^{(o_{1:N})}$:
		\begin{equation}\label{eq:Gsize2'}
			\left|G^{(o_{1:N})}\right| \ge  \left( 1-N \left(\frac{1}{2^n+1}\right)^{1/3}\right)(4^n-1).
		\end{equation}
		We are now ready to upper bound $\mr{TV}(p_1,p_2)$ from Eq.~\eqref{eq:TV1'}. The strategy is to divide the sum over all Pauli operators into $G^{(o_{1:N})}$ and ${\sf P}^n\backslash G^{(o_{1:N})}$. All terms within the former group are small thanks to the definition of $G^{(o_{1:N})}$; Terms within the latter group could be large, but the total number of them are small. Combining these two gives a pretty good upper bound on $\mr{TV}(p_1,p_2)$. 
		In math,
		\begin{equation}
			\begin{aligned}
				\mr{TV}(p_1,p_2) &= \frac{1}{4^n-1}\sum_{\substack{
					o_{1:N}~\text{s.t.}\\
					p_1(o_{1:N})\ge p_2(o_{1:N})
				}}
				p_1(o_{1:N})
				\left(
					\sum_{a\in  G^{(o_{1:N})}} + 
					\sum_{\substack{
						a\in  \mbb Z^{2n}_2\backslash G^{(o_{1:N})},\\
					a\ne 0
					}}
				\right)
				\left(1-\prod_{i=1}^{N}\left(
					1 + \xi_a^{o_{<i}}[M^{o_{<i}}] \bra{B_{o_i}^{o_{<i}}}P_a\ket{B_{o_i}^{o_{<i}}}
				\right)\right)\\
				&\le \frac{1}{4^n-1}\sum_{\substack{
					o_{1:N}~\text{s.t.}\\
					p_1(o_{1:N})\ge p_2(o_{1:N})
				}}p_1(o_{1:N})
				\left(
					\left|G^{(o_{1:N})}\right|\times\left( 1-\left(1- \left(\frac{1}{2^n+1}\right)^{1/3}  \right)^N \right)
					+
					\left(4^n-1-\left|G^{(o_{1:N})}\right|\right)
				\right)\\
				&\le \sum_{\substack{
					o_{1:N}~\text{s.t.}\\
					p_1(o_{1:N})\ge p_2(o_{1:N})
				}}p_1(o_{1:N})
				\left(
					N\left( \frac{1}{2^n+1} \right)^{1/3} + 
					N\left( \frac{1}{2^n+1} \right)^{1/3}
				\right)\\
				& \le 2N\left( \frac{1}{2^n+1} \right)^{1/3}.
			\end{aligned}
		\end{equation}
		The first inequality applies the bound $|\xi_a^{o_{<i}}[M^{o_{<i}}]|\le 1$. 
		Besides, the first sum uses the bound from the definition of $G^{(o_{1:N})}$,
		and the second sum is bounded using $\left| \bra{B_{o_i}^{o_{<i}}} P_a \ket{B_{o_i}^{o_{<i}}} \right| \le 1$. 
		The second inequality uses the bounds $\left| G^{(o_{1:N})} \right| \le 4^n - 1$ for the first sum and Eq.~\eqref{eq:Gsize2'} for the second sum, as well as the fact that $1-(1-x)^N\le Nx$ for all $0\le x\le 1$. 
		Now we have obtained the claimed bound $\mr{TV}(p_1,p_2) = O(N2^{-n/3})$, and hence complete the proof of Theorem~\ref{th:no_ancilla}.
	\end{proof}
	
	\subsection{A lower bound for the most general entangled strategies}
	
		In this section, we prove a lower bound of $\Omega(n)$ for the fully entangled estimation strategies, which is the most general measurement strategies one can do to learn an unknown channel even with the help of quantum computers, see Fig.~\ref{fig:model} (a).
	Note that, since we assume there is an unlimited amount of quantum memory, we can without loss of generality eliminate any intermediate measurements, and only conduct one joint measurement after sequentially processing all $N$ samples of the channel.
	\begin{theorem}\label{th:entlower}
		For any fully entangled measurement protocols that give an estimate $\widehat{\bm{\lambda}}$ for the Pauli eigenvalues $\bm\lambda$ of an arbitrary unknown $n$-qubit Pauli channel $\Lambda$ such that
		\begin{equation}
			|\widehat\lambda_a - {\lambda_a}|< \frac12,\quad \forall a\in\mbb Z_2^{2n}
		\end{equation}
		holds with high probability, the number of samples of $\Lambda$ required is at least $\Omega(n)$.
	\end{theorem}

	\begin{proof}
		We first show that, by using a technique known as \emph{teleportation stretching}~\cite{pirandola2017fundamental,pirandola2019fundamental}, any fully entangled measurement protocols for $N$ copies of an arbitrary Pauli channel $\Lambda$ can be simulated by a joint measurement on $N$ copies of the Choi state $J_\Lambda$ which is defined as
		\begin{equation}
		\begin{aligned}
			J_\Lambda \coleq&~ \Lambda\otimes\mathds 1(\ketbra{\Psi^+}{\Psi^+})\\
			=&~ \frac{1}{4^n}\sum_{a\in\mbb Z_{2}^{2n}}\lambda_a P_a\otimes P_a^T.
		\end{aligned}
		\end{equation}
		This follows from the existence of a quantum channel $\mc T$ such that $\Lambda(\rho) = \mc T(\rho\otimes J_\Lambda)$ holds for all Pauli channels $\Lambda$. 
		One possible construction of $\mc T$ is shown in Fig.~\ref{fig:telesim} (see Ref.~\cite{bennett1996mixed}).
		In word, one first applies a Bell measurement on the input state $\rho$ and half of the Choi state $J_\Lambda$. Then, conditioned on the Bell measurement outcome $\ket{\Psi_b}$, one applies a Pauli correction $P_b$ on the other half of $J_\Lambda$, which will then be equal to $\Lambda(\rho)$.
		Indeed, the post-measurement state conditioned on Bell measurement outcome $b$ is
		\begin{equation}
		\begin{aligned}
		    \rho_b &\propto
		    \bra{\Psi_b}_{AB}\rho^A\otimes J_\Lambda^{BC}\ket{\Psi_b}_{AB}\\
		    &\propto\sum_{a\in\mbb Z_2^{2n}}\lambda_a \bra{\Psi_b}\rho\otimes P_{a}\ket{\Psi_b}\otimes P_a^T\\
		    &=\frac{1}{2^n}\sum_{a\in\mbb Z_2^{2n}}\lambda_a (-1)^\expval{a,b}\Tr(\rho P_a)P_a.
		\end{aligned}
		\end{equation}
		After applying the Pauli correction, the state becomes
		\begin{equation}
		    P_b\rho_b P_b = \frac{1}{2^n}\sum_{a\in\mbb Z_2^{2n}}\lambda_a\Tr(\rho P_a)P_a = \Lambda(\rho),
		\end{equation}
		which justify the relation $\Lambda(\rho) = \mc T(\rho\otimes J_\Lambda)$.
		
		\begin{figure}[!htp]
			\centering
			\includegraphics[width = 0.5\textwidth]{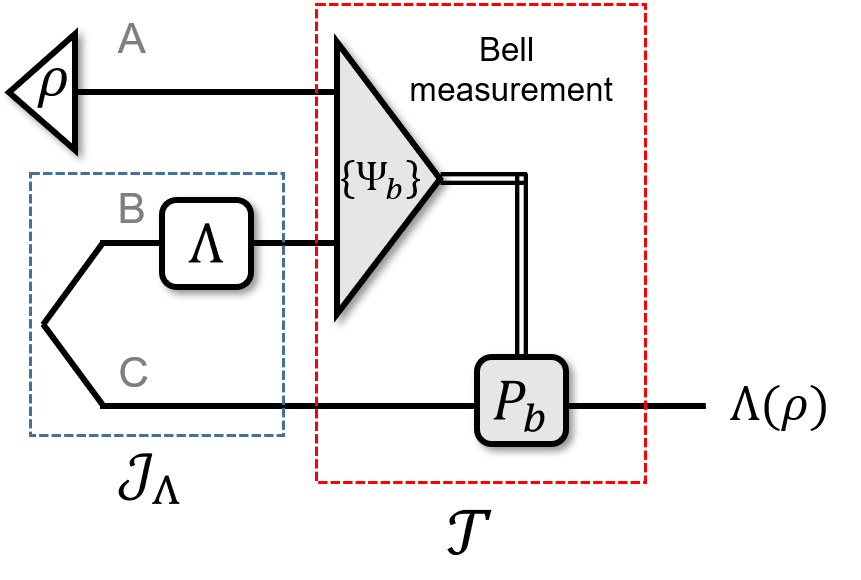}
			\caption{
				Construction of the teleportation simulation channel $\mc T$.
			}
			\label{fig:telesim}
		\end{figure}

		With the help of teleportation stretching, one can reduce any measurement protocols for $N$ copies of $\Lambda$ to a single POVM measurement on $N$ copies of $J_\Lambda$, as shown in Fig.~\ref{fig:reduction}.

		\begin{figure}[!htp]
			\centering
			\includegraphics[width = 0.5\textwidth]{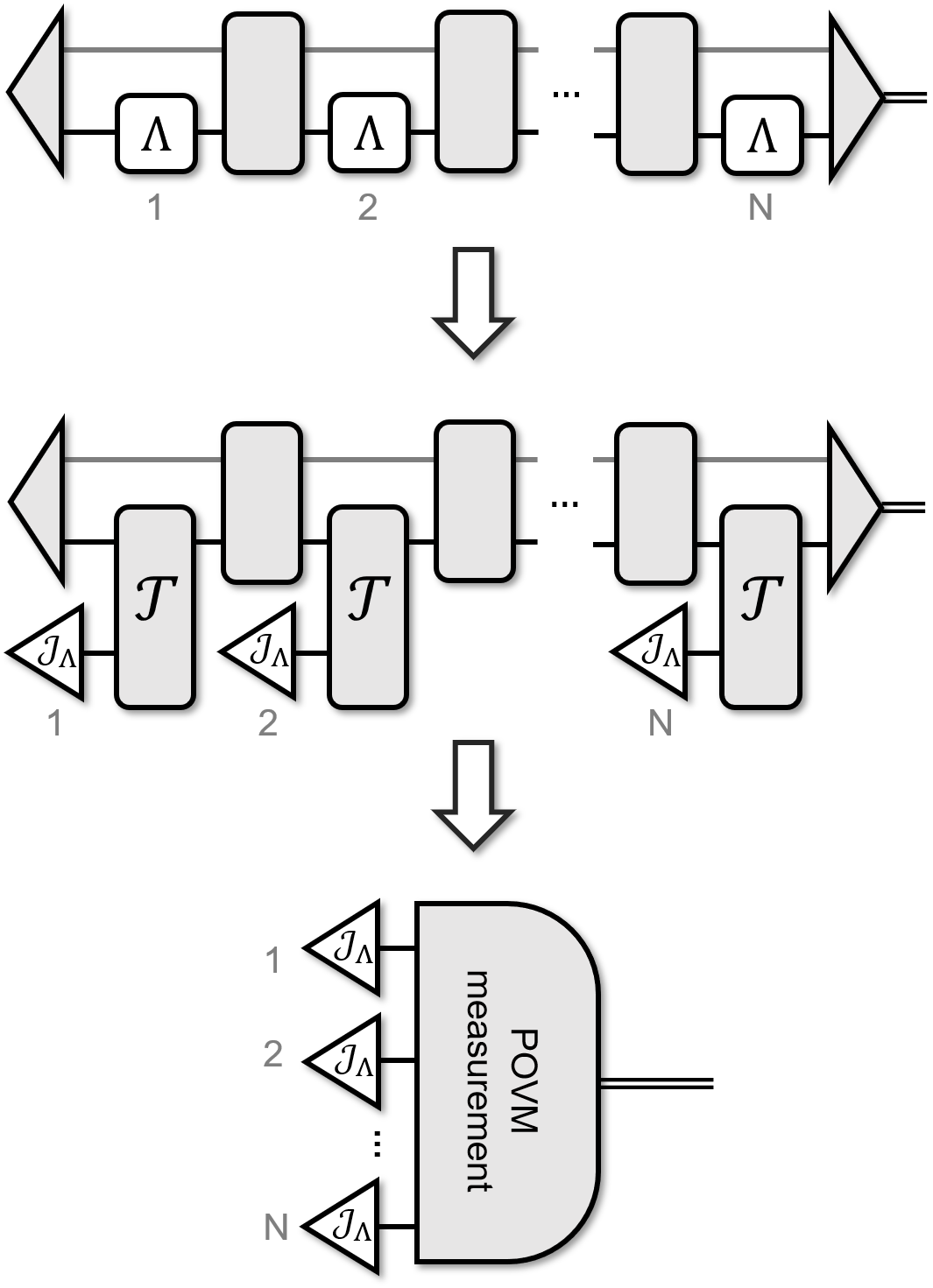}
			\caption{
				Teleportation stretching~\cite{pirandola2017fundamental,pirandola2019fundamental}: simulation of an arbitrary entangled measurement on $\Lambda$ by a single measurement on $J_\Lambda$. One just needs to simulate every application of $\Lambda(\cdot)$ by $\mc T((\cdot)\otimes\rho)$. The whole measurement protocol then become a joint POVM measurement on $N$ copies of $J_\Lambda$ with no adaptivity.  
			}
			\label{fig:reduction}
		\end{figure}

		\medskip

		Now, recall the communication task defined in Sec.~\ref{app:lower1}, where Alice and Bob share the following ``codebook''
		\begin{equation}
			(a,s)\in\{1,\cdots,4^n-1\}\times\{\pm 1\} ~\longrightarrow~ \Lambda_{(a,s)}(\cdot) = \frac{1}{2^n}\left(I\Tr(\cdot) + s P_a\Tr(P_a(\cdot))\right),
		\end{equation}
		and Alice randomly picks one possible $(a,s)$ and send $N$ copies of $\Lambda_{(a,s)}$ to Bob. 
		If there exists a fully entangled estimation protocol using $N$ samples and satisfying the assumption of Theorem~\ref{th:entlower}, Bob can determine Alice's choice of $(a,s)$ with high probability.
		According to Fano's inequality, the mutual information between the random variable pair $(a,s)$ and Bob's measurement result $o$ has the following lower bound
	\begin{equation}
		I((a,s):o) \ge \Omega(\log2(4^n-1)) = \Omega(n).
	\end{equation}
	One the other hand, since any measurement Bob conducts can be simulated by a measurement on $N$ copies of $J_{\Lambda_{(a,s)}}$, Holevo's theorem~\cite{holevo1973bounds,wilde2013quantum} can be apply to $I((a,s):o)$. We have
	\begin{equation}
		\begin{aligned}
			I((a,s):o) &\le S\left(\mathop{\mbb E}_{(a,s)} J^{\otimes N}_{\Lambda_{(a,s)}} \right) - \mathop{\mbb E}_{(a,s)}S\left( J^{\otimes N}_{\Lambda_{(a,s)}} \right)\\
			&=S\left(\mathop{\mbb E}_{(a,s)} J^{\otimes N}_{\Lambda_{(a,s)}} \right) - N \mathop{\mbb E}_{(a,s)}S\left( J_{\Lambda_{(a,s)}} \right)\\
			&\le 2nN - (2n-1)N\\
			& = N.
		\end{aligned}
	\end{equation}
	In the third line, the first term is a trivial upper bound for the von Neumann entropy on a $2^{2nN}$-dimensional Hilbert space. The second term uses the observation that 
	\begin{equation}
	    J_{\Lambda_{(a,s)}}= \frac{1}{4^n}(I\otimes I + s P_a\otimes P_a^{T})
	\end{equation}
	is a maximally mixed state on a $2^{2n-1}$-dimensional Hilbert space, thus $S(J_{\Lambda_{(a,s)}}) = 2n-1$.
	This yields the lower bound $N = \Omega(n)$.
	\end{proof}

\section{SPAM-robust ancilla-assisted Pauli gate benchmarking protocols}

In the main text, we have described an ancilla-assisted Pauli channel estimation protocol which provides exponential advantages over any ancilla-free protocols. 
Several issues need to be addressed before applying these protocols to a practical quantum noise characterization setting.
Firstly, in most cases, we do not have a ``black-box'' access to the Pauli noise channels of interest. Instead, they are often attached with some applied quantum gates. We must consider the effect of such gates in our protocol.
Secondly, the state preparation and measurement (SPAM) process would inevitably suffer from error. We would like to minimize the effect of such errors.

A recent breakthrough by Flammia and Wallman~\cite{flammia2020efficient} provides a way to address these issues, using ideas from \textit{randomized benchmarking}~(see \cite{helsen2020general} and references therein). Their task is to benchmark the Pauli error rates of the Pauli gate set. By concatenating $m+1$ layers of random Pauli gates, they effectively obtain the $m$th power of the Pauli twirl for the noise channel (plus a single Pauli correction gate at the end), under a gate-independent, time-stationary, and Markovian (GTM) noise assumption.
They then describe a protocol to estimate the quantity $A_a\lambda_a^{m}$ for all $a\in\mbb Z_2^{2n}$, where $\bm \lambda \coleq \{\lambda_a\}_a$ is the Pauli eigenvalues of the noise channel of interest, and $A_a$ is a SPAM related constant that is independent of $m$. By repeating this estimation procedure for different concatenating length $m$ and applying a single-exponential fitting of $A_a\lambda_a^m$ for each $a\in\mbb Z_2^{2n}$, one obtained an SPAM-robust estimation for $\bm\lambda$.

Importantly, the protocol in~\cite{flammia2020efficient} is ancilla-free, which means an exponential number of measurements is necessary to approximate $\bm\lambda$ to small error in $l_\infty$ distance, according to our Theorem~\ref{th:lower_main}~(C). In this section, we explain how the ancilla-assisted Pauli channel estimation protocol described in the main text can be extended to this Pauli gates benchmarking setting, which is able to estimate $\bm\lambda$ exponentially more sample-efficiently as well as being SPAM-robust.
The new protocol uses gate concatenation and single-exponential fitting, and can be viewed as a generalization of the methods in~\cite{flammia2020efficient}.

\medskip

Let us start by defining the task and specifying our assumptions.
The task is to characterize the noise of the $n$-qubit Pauli gate set. We assume the noise satisfies the GTM condition, which means every noisy implementation of Pauli gates can be written as
\begin{equation}\label{eq:GTM}
    \widetilde{\mc P}_a = \mc P_a \Lambda_G,\quad a\in \mbb Z_2^{2n},
\end{equation}
for an $a$-independent quantum channel $\Lambda_G$. 
The calligraphic $\mc P_a$ is the channel representation of the Pauli gate $P_a$, \textit{i.e.}, $\mc P_a(\cdot) \coleq P_a(\cdot)P_a$.
Our specific goal is to estimate the Pauli eigenvalues of the Pauli twirl of $\Lambda_G$, which is defined as
\begin{equation}\label{eq:PauliTwirl}
    \Lambda\coleq \frac{1}{4^n}\sum_{a\in \mbb Z_2^{2n}}\mc P_a\Lambda_G\mc P_a.
\end{equation}
In addition, we assume there to be an $n$-qubit ancillary systems that can be entangled with the main system. 
The ancilla will basically be used as a quantum memory~(see Fig.~\ref{fig:spam} in the main text).
A crucial assumption we need is that, the noise on the ancilla is negligible except for the entangled state preparation and measurement procedure.
In other word, the noise channel on the ancilla is independent of the concatenating length on the main system.
In practice, this requires (1) the crosstalk between the ancilla and the main system is negligible when applying gates only on the main system, and (2) the coherence time of the ancilla is much longer compared to the time of applying gates on the main system.
We expect that these assumptions can be satisfied by \textit{e.g.}, a near-term ion trap platform (see~\cite{wright2019benchmarking,pino2021demonstration}).
Ion trap system typically has very long coherence time; Besides, after preparing the entangled state, one can shuttle the ions to separate the ancilla and the main system during gate applications and shuttle them back for the entangled measurement~\cite{pino2021demonstration}.
This step can minimize the crosstalk, and the errors introduced there can be viewed as SPAM error (independent of the concatenating length) so our protocol will be naturally robust against them.

The benchmarking protocol is described in Algorithm~\ref{alg:spam}. In the following, we will prove the correctness and give a rough analysis on the sample efficiency. A more rigorous analysis on the sample complexity, optimization of the concatenating length ${\sf M}$ and the number of repetitions $R$, and other aspects of the protocol are left for future research. 
We also remark that, for simplicity, we focus on the case where $k=n$ ancillary qubits are available. For a restricted number of ancillary qubits $0<k<n$, one can also design a similar benchmarking protocol by hybridizing the $k=n$ protocol here and the $k=0$ protocol in~\cite{flammia2020efficient}. We omit the details about this hybrid protocol.

\begin{theorem}\label{th:spam}
    Given the aforementioned two assumptions about the noise model, the estimator given at Line~6 in Algorithm~\ref{alg:spam} satisfies
    \begin{equation}
        \mathop{\mbb E}\left[ \widehat{F}_a^{(k)}(m) \right]
        = A_a\lambda_a^m,\quad \forall a\in\mbb Z_2^{2n},
    \end{equation}
    where $A_a$, defined in Eq.~\eqref{eq:constA}, is a noise-dependent constant that is independent of $m$.
\end{theorem}

Theorem~\ref{th:spam} guarantees that, given sufficiently many concatenating length $m$ and circuit samples $R$, Algorithm~\ref{alg:spam} can indeed converge to the true Pauli eigenvalues $\bm\lambda$, in a SPAM-error resilient manner. 
Since $\widehat{F}_a^{(k)}(m)$ takes value from $\{1,-1\}$, Hoeffding's bound says that a constant number of samples is enough to estimate its expectation to constant additive error with $1-o(1)$ success probability, for any specific $a$ and $m$.
By the union bound, $\mc O(n)$ samples are enough for this to hold for all $a\in\mbb Z_2^{2n}$ simultaneously with high probability.
If we further assume that the noise is weak, so that both $A_a$ and $\lambda_a$ are lower bounded by some constant, then a constant number of $m$ and the above-achieved constant additive precision is enough for a small final estimation error for $\bm\lambda$ in $l_\infty$ distance, which implies a total sample complexity of $\mc O(n)$. 
Therefore, Algorithm~\ref{alg:spam} is indeed exponentially more sample-efficient than the ancilla-free protocol~\cite{flammia2020efficient} for this specific task.

\begin{figure}[htb]
\begin{algorithm}[H]
		\caption{SPAM-robust ancilla-assisted Pauli gate benchmarking}
		\label{alg:spam}
		\begin{algorithmic}[1]
			\Require
			(1) List of concatenating length $\sf M$.
			(2) Number of Repetitions $R$.
			(3) Noisy implementation of Pauli gates $\widetilde{\mc P}_a = \mc P_a\Lambda_G$.
			\Ensure
			SPAM-robust estimates $\widehat{\bm\lambda}$ for the Pauli eigenvalues of $\Lambda$ as defined in 
			Eq.~\eqref{eq:PauliTwirl}.
            \For{$m\in {\sf M}$}
            \For{$k = 1~\textbf{to}~R$}
            \State Prepare the (noisy) Bell state $\ket{\widetilde{\Psi}^+}$ between the ancillary system and the main system.
            \State Sequentially apply $m+1$ random (noisy) $n$-qubit Pauli gates $\{\widetilde{P}_{a_t}\}_{t=0}^{m}$ to the main system.
            \State Apply the (noisy) Bell measurement $\{\ket{\widetilde{\Psi}_v}\}_v$ with outcome $v$.
            \State $\widehat{F}^{(k)}_{a}(m)\coleq(-1)^{\expval{a,v}+\sum_{t=0}^m\expval{a,a_t}}$~\textbf{for all}~$a\in\mbb Z_2^{2n}$.
            \EndFor
            \EndFor
            \For{$a\in\mbb Z_2^{2n}$}
            \State $\widehat{F}_{a}(m)\coleq \frac{1}{R}\sum_{k=1}^R\widehat{F}^{(k)}_{a}(m)$.
            \State Fit $\widehat{F}_{a}(m)$ to the single-exponential decay model $\widehat{A}_a\widehat{\lambda}_a^m$.
            \EndFor
            \State\Return $\widehat{\bm\lambda}\coleq\{\widehat{\lambda}_a\}_a$.
		\end{algorithmic}
\end{algorithm}
\end{figure}

Before presenting the proof of Theorem~\ref{th:spam}, we introduce the \emph{Pauli-transfer-matrix} (PTM) representation to simplify notations.
A linear operator $O$ acting on a $2^n$-dimensional Hilbert space can be viewed as a vector in a $4^n$-dimensional Hilbert space. We denote this vectorization of $O$ as $\lket{O}$ and the corresponding Hermitian conjugate as $\lbra{O}$. The inner product within this space is the \emph{Hilbert-Schmidt} product defined as $\lbraket{A}{B}\coleq \Tr(A^\dagger B)$. The \emph{normalized Pauli operators} $\{\sigma_a\coleq P_a/\sqrt{2^{n}},~a\in\mbb Z_2^{2n}\}$ forms an orthonormal basis for this space. 
In the PTM representation, a superoperator (\textit{i.e.}, quantum channel) becomes an operator acting on the $4^n$-dimensional Hilbert space, sometimes called the Pauli transfer operator. 
Explicitly, we have $\lket{\Lambda(\rho)} = \Lambda^\ptm\lket{\rho} \equiv \Lambda\lket{\rho}$,
where we use the same notation to denote a channel and its Pauli transfer operator, which should be clear from the context.
Specifically, a general Pauli channel $\Lambda$ has the follwing Pauli transfer operator
$$
\Lambda = \sum_{a\in\mbb Z_2^{2n}} \lambda_a\lketbra{\sigma_a}{\sigma_a},
$$
where $\{\lambda_a\}_a$ are the Pauli eigenvalues. It is also obvious that the $m$-th power of $\Lambda$ is
$$
\Lambda^m = \sum_{a\in\mbb Z_2^{2n}} \lambda^m_a\lketbra{\sigma_a}{\sigma_a}.
$$

\noindent Using the PTM representation, the constant $A_a$ in Theorem~\ref{th:spam} is defined as
\begin{equation}\label{eq:constA}
    A_a \coleq \sum_{v\in\mbb Z_2^{2n}}(-1)^\expval{a,v} \lbra{\widetilde{\Psi}_v}\mathds 1\otimes(\lket{\sigma_a}\lbra{\sigma_a}\Lambda_G) \lket{\widetilde{\Psi}^+},
\end{equation}
where $\lket{\widetilde{\Psi}^+}$ is just the PTM representation for the density matrix of the (noisy) Bell state %
$\widetilde{\Psi}^+$.
Same for $\lket{\widetilde{\Psi}_b}$.
One can verify that $A_a=1$ for the noiseless case (where there is no SPAM error and $\Lambda_G=\mathds 1$).

\begin{proof}[Proof of Theorem~\ref{th:spam}]
    The following proof is a generalization of~\cite[Proposition~5]{flammia2020efficient} and we borrow some of their presentations.
    Consider the probability that a specific sequence of Pauli gates $\{P_{a_t}\}_{t=0}^m$ is sampled (Line~4,  Alg.~\ref{alg:spam}) and the Bell measurement outcome is $v$ (Line~5, Alg.~\ref{alg:spam}),
    \begin{equation}\label{eq:spam_dis1}
    \begin{aligned}
        \Pr(a_0,\cdots,a_m,v) &= 
        \frac{1}{4^{n(m+1)}} \lbra{\widetilde\Psi_v} \mathds 1\otimes\left(
        \mc P_{a_m}\Lambda_{G}\cdots \mc P_{a_1}\Lambda_{G}\mc P_{a_0}\Lambda_{G}
        \right) \lket{\widetilde\Psi^+}
        \\
        &=\frac{1}{4^{n(m+1)}} \lbra{\widetilde\Psi_v} \mathds 1\otimes\left(
          \prod_{t=m}^0 \mc P_{a_t}\Lambda_{G}
        \right) \lket{\widetilde\Psi^+},
    \end{aligned}
    \end{equation}

	\noindent Here we use the assumption that the noise on the ancilla is negligible except for the state preparation and measurement part. We can absorb the noise channel on the ancilla into the SPAM error that is independent of the concatenating length $m$. That is why we can have an $\mathds 1$ on the ancillary system.

    The distribution can be rewritten as
    \begin{equation}\label{eq:spam_dis2}
    \begin{aligned}
        \Pr(a_0,\cdots,a_m,v) 
        &=\frac{1}{4^{n(m+1)}} \lbra{\widetilde\Psi_v} \mathds 1\otimes P_{a_{m}'}\left(
          \prod_{t=m-1}^0 \mc P_{a_t'}\Lambda_{G}\mc P_{a_t'}
        \right) \Lambda_{G}\lket{\widetilde\Psi^+},
    \end{aligned}
    \end{equation}
    where we define $a_t'\coleq \sum_{k=0}^t a_k$ (bit-wise modulo 2 sum). 
    Taking this change-of-variables and averaging over $\{a_0,~\cdots,~a_{m-1}\}$, we get
    \begin{equation}\label{eq:spam_dis3}
    \begin{aligned}
        \Pr(a'_m,v) 
        &=\frac{1}{4^{n}} \lbra{\widetilde\Psi_v} \mathds 1\otimes \mc P_{a_{m}'}\left(
          \prod_{t=m-1}^0 \mathop{\mbb E}_{a_{t}'\in\mbb Z_2^{2n}} \mc P_{a_t'}\Lambda_{G}\mc P_{a_t'}
        \right) \Lambda_{G}\lket{\widetilde\Psi^+}\\
        &=\frac{1}{4^{n}} \lbra{\widetilde\Psi_v} \mathds 1\otimes \mc P_{a_{m}'}
          \Lambda^m
        \Lambda_{G}\lket{\widetilde\Psi^+}\\
        &=\frac{1}{4^{n}}\sum_{a\in\mbb Z_2^{2n}}\lambda_a^m
        \lbra{\widetilde\Psi_v} \mathds 1\otimes (\mc P_{a_{m}'}
          \lketbra{\sigma_a}{\sigma_a}
        \Lambda_{G})\lket{\widetilde\Psi^+}\\
        &=\frac{1}{4^{n}}\sum_{a\in\mbb Z_2^{2n}}(-1)^\expval{a,a_m'}\lambda_a^m
        \lbra{\widetilde\Psi_v} \mathds 1\otimes (
          \lketbra{\sigma_a}{\sigma_a}
        \Lambda_{G})\lket{\widetilde\Psi^+}.
    \end{aligned}
    \end{equation}
    Define $z\coleq v+a_m'$, the marginal distribution of $z$ is
    \begin{equation}\label{eq:spam_dis4}
    \begin{aligned}
        \Pr(z) 
        &= \sum_{v\in\mbb Z_2^{2n}}\Pr(v+z,v)\\
        &= \frac{1}{4^{n}}\sum_{a\in\mbb Z_2^{2n}}(-1)^\expval{a,z}
        \lambda_a^m
        \sum_{v\in\mbb Z_2^{2n}}(-1)^\expval{a,v}
        \lbra{\widetilde\Psi_v} \mathds 1\otimes (
          \lketbra{\sigma_a}{\sigma_a}
        \Lambda_{G})\lket{\widetilde\Psi^+}\\
        &= \frac{1}{4^{n}}\sum_{a\in\mbb Z_2^{2n}}(-1)^\expval{a,z}
        \lambda_a^m A_a.
    \end{aligned}
    \end{equation}
    Apply the inverse Walsh-Hadamard transform, we obtain
    \begin{equation}
        A_a\lambda_a^m = \sum_{z\in\mbb Z_2^{2n}}(-1)^\expval{a,z}\Pr(z).
    \end{equation}
    Therefore, $(-1)^\expval{a,z} = (-1)^{\expval{a,v}+\sum^m_{t=0}\expval{a,a_t}}$ is an unbiased estimator for $A_a\lambda_a^m$. In other word,
    \begin{equation}
        \mathop{\mbb E}\left[\widehat{F}_a^{(k)}(m)\right] \equiv
        \mathop{\mbb E}\left[ (-1)^\expval{a,z} \right] = 
        A_a\lambda_a^m.
    \end{equation}
    This is exactly the claim of Theorem~\ref{th:spam}.
\end{proof}

\end{appendix}

\end{document}